\documentclass[draftcls,onecolumn, 12pt]{IEEEtran}

\usepackage{silence}
\usepackage{xparse}
\usepackage{listings}
\usepackage{babel}
\usepackage{bbm} 
\usepackage{lipsum}
\usepackage{forest}
\usepackage{upgreek}
\usepackage[LGR,T1]{fontenc}
\usepackage{ulem} 

\DeclareSymbolFont{ttgreek}{LGR}{cmtt}{m}{n}

\DeclareMathSymbol{\ttalpha}{\mathord}{ttgreek}{`a}
\DeclareMathSymbol{\ttAlpha}{\mathord}{ttgreek}{`A}
\DeclareMathSymbol{\ttmu}{\mathord}{ttgreek}{`m}
\DeclareMathSymbol{\tttau}{\mathord}{ttgreek}{`t}
\usepackage{alphabeta}

\usepackage[T1]{fontenc}
\usepackage[LGR,T1]{fontenc}
\usepackage[utf8]{inputenc}
\usepackage{mathtools}

\usepackage{amssymb,mathrsfs}
\usepackage{amsthm}
\usepackage{bm}
\usepackage{scalerel}
\usepackage{nicefrac}
\usepackage{microtype} 
\usepackage[shortlabels]{enumitem}
\usepackage{graphicx}
\usepackage{epstopdf}
\DeclareGraphicsExtensions{.eps,.png,.jpg,.pdf}

\usepackage{url}
\usepackage{colortbl}
\usepackage{booktabs}
\usepackage{multirow}
\usepackage{colortbl,xcolor}
\usepackage{soul}
\usepackage{xparse,xstring}
\usepackage{calc}
\usepackage{etoolbox}

\makeatletter
\@ifpackageloaded{natbib}{
	\relax
}{
	\usepackage{cite}
}
\makeatother

\usepackage{array}
\newcolumntype{L}[1]{>{\raggedright\let\newline\\\arraybackslash\hspace{0pt}}m{#1}}
\newcolumntype{C}[1]{>{\centering\let\newline\\\arraybackslash\hspace{0pt}}m{#1}}
\newcolumntype{R}[1]{>{\raggedleft\let\newline\\\arraybackslash\hspace{0pt}}m{#1}}

\makeatletter
\let\MYcaption\@makecaption
\makeatother
\usepackage[font=footnotesize]{subcaption}
\makeatletter
\let\@makecaption\MYcaption
\makeatother

\usepackage{glossaries}
\makeatletter
\sfcode`\.1006

\let\oldgls\gls
\let\oldglspl\glspl

\newcommand\fussy@ifnextchar[3]{%
	\let\reserved@d=#1%
	\def\reserved@a{#2}%
	\def\reserved@b{#3}%
	\futurelet\@let@token\fussy@ifnch}
\def\fussy@ifnch{%
	\ifx\@let@token\reserved@d
		\let\reserved@c\reserved@a
	\else
		\let\reserved@c\reserved@b
	\fi
	\reserved@c}

\renewcommand{\gls}[1]{%
\oldgls{#1}\fussy@ifnextchar.{\@checkperiod}{\@}}
\renewcommand{\glspl}[1]{%
\oldglspl{#1}\fussy@ifnextchar.{\@checkperiod}{\@}}

\newcommand{\@checkperiod}[1]{%
	\ifnum\sfcode`\.=\spacefactor\else#1\fi
}

\robustify{\gls}
\robustify{\glspl}
\makeatother

\newacronym{wrt}{w.r.t.}{with respect to}
\newacronym{RHS}{R.H.S.}{right-hand side}
\newacronym{LHS}{L.H.S.}{left-hand side}
\newacronym{iid}{i.i.d.}{independent and identically distributed}
\newacronym{SOTA}{SOTA}{state-of-the-art}

\usepackage{float}

\ifx\notloadhyperref\undefined
	\ifx\loadbibentry\undefined
		\usepackage[hidelinks,hypertexnames=false]{hyperref}
	\else
		\usepackage{bibentry}
		\makeatletter\let\saved@bibitem\@bibitem\makeatother
		\usepackage[hidelinks,hypertexnames=false]{hyperref}
		\makeatletter\let\@bibitem\saved@bibitem\makeatother
	\fi
\else
	\ifx\loadbibentry\undefined
		\relax
	\else
		\usepackage{bibentry}
	\fi
\fi

\usepackage[capitalize]{cleveref-forward}

\crefname{equation}{}{}
\Crefname{equation}{}{}
\crefname{claim}{claim}{claims}
\crefname{step}{step}{steps}
\crefname{line}{line}{lines}
\crefname{condition}{condition}{conditions}
\crefname{dmath}{}{}
\crefname{dseries}{}{}
\crefname{dgroup}{}{}
\crefname{page}{page}{pages}

\crefname{Problem}{Problem}{Problems}
\crefformat{Problem}{Problem~#2#1#3}
\crefrangeformat{Problem}{Problems~#3#1#4 to~#5#2#6}

\crefname{Theorem}{Theorem}{Theorems}
\crefname{Corollary}{Corollary}{Corollaries}
\crefname{Proposition}{Proposition}{Propositions}
\crefname{Lemma}{Lemma}{Lemmas}
\crefname{Definition}{Definition}{Definitions}
\crefname{Example}{Example}{Examples}
\crefname{Assumption}{Assumption}{Assumptions}
\crefname{Remark}{Remark}{Remarks}
\crefname{Rem}{Remark}{Remarks}
\crefname{remarks}{Remarks}{Remarks}
\crefname{Appendix}{Appendix}{Appendices}
\crefname{Supplement}{Supplement}{Supplements}
\crefname{Exercise}{Exercise}{Exercises}
\crefname{TheoremA}{Theorem}{Theorems}
\crefname{CorollaryA}{Corollary}{Corollaries}
\crefname{PropositionA}{Proposition}{Propositions}
\crefname{LemmaA}{Lemma}{Lemmas}
\crefname{DefinitionA}{Definition}{Definitions}
\crefname{ExampleA}{Example}{Examples}
\crefname{RemarkA}{Remark}{Remarks}
\crefname{AssumptionA}{Assumption}{Assumptions}
\crefname{ExerciseA}{Exercise}{Exercises}
\crefname{algorithm}{Algorithm}{Algorithms}
\crefname{figure}{Fig.}{Figs.}
\crefname{table}{Table}{Tables}
\crefname{section}{Section}{Sections}
\crefname{subsection}{Section}{Sections}
\crefname{subsubsection}{Section}{Sections}

\usepackage{crossreftools}
\ifx\notloadhyperref\undefined
\else
	\relax
\fi

\usepackage{algorithm}
\usepackage{algpseudocode}

\ifx\loadbreqn\undefined
	\relax
\else
	\usepackage{breqn}
\fi

\makeatletter
\def\cleartheorem#1{%
    \expandafter\let\csname#1\endcsname\relax
    \expandafter\let\csname c@#1\endcsname\relax
}
\def\clearthms#1{ \@for\tname:=#1\do{\cleartheorem\tname} }
\makeatother

\ifx\renewtheorem\undefined
	\ifx\useTheoremCounter\undefined
		\newtheorem{Theorem}{Theorem}
		\newtheorem{Corollary}{Corollary}
		\newtheorem{Proposition}{Proposition}
		
	\else
		\newtheorem{Theorem}{Theorem}

	\fi

	\newtheorem{Example}{Example}
	
	\newtheorem{Assumption}{Assumption}

	\newtheorem{LemmaA}{Lemma}[section]

\fi

\theoremstyle{remark}

\theoremstyle{plain}

\newcommand{\qednew}{\nobreak \ifvmode \relax \else
		\ifdim\lastskip<1.5em \hskip-\lastskip
			\hskip1.5em plus0em minus0.5em \fi \nobreak
		\vrule height0.75em width0.5em depth0.25em\fi}



\newcommand{\nn}{\nonumber\\ }

\NewDocumentCommand{\movedownsub}{e{^_}}{%
	\IfNoValueTF{#1}{%
		\IfNoValueF{#2}{^{}}
	}{%
		^{#1}
	}%
	\IfNoValueF{#2}{_{#2}}
}

\let\latexchi\chi
\RenewDocumentCommand{\chi}{}{\latexchi\movedownsub}

\newcommand{\Real}{\mathbb{R}}

\newcommand{\calE}{\mathcal{E}}
\newcommand{\calF}{\mathcal{F}}

\newcommand{\calJ}{\mathcal{J}}
\newcommand{\calK}{\mathcal{K}}

\newcommand{\calT}{\mathcal{T}}

\newcommand{\calV}{\mathcal{V}}

\newcommand{\calZ}{\mathcal{Z}}

\DeclareSymbolFont{ttgreek}{LGR}{cmtt}{m}{n}
\DeclareMathSymbol{\ttalpha}{\mathord}{ttgreek}{`a}
\DeclareMathSymbol{\ttbeta}{\mathord}{ttgreek}{`b}
\DeclareMathSymbol{\ttgamma}{\mathord}{ttgreek}{`g}
\DeclareMathSymbol{\ttdelta}{\mathord}{ttgreek}{`d}
\DeclareMathSymbol{\ttepsilon}{\mathord}{ttgreek}{`e}
\DeclareMathSymbol{\ttzeta}{\mathord}{ttgreek}{`z}
\DeclareMathSymbol{\tteta}{\mathord}{ttgreek}{`h}
\DeclareMathSymbol{\tttheta}{\mathord}{ttgreek}{`j}
\DeclareMathSymbol{\ttiota}{\mathord}{ttgreek}{`i}
\DeclareMathSymbol{\ttkappa}{\mathord}{ttgreek}{`k}
\DeclareMathSymbol{\ttlambda}{\mathord}{ttgreek}{`l}
\DeclareMathSymbol{\ttmu}{\mathord}{ttgreek}{`m}
\DeclareMathSymbol{\ttnu}{\mathord}{ttgreek}{`n}
\DeclareMathSymbol{\ttxi}{\mathord}{ttgreek}{`x}
\DeclareMathSymbol{\ttomicron}{\mathord}{ttgreek}{`o}
\DeclareMathSymbol{\ttpi}{\mathord}{ttgreek}{`p}
\DeclareMathSymbol{\ttrho}{\mathord}{ttgreek}{`r}
\DeclareMathSymbol{\ttsigma}{\mathord}{ttgreek}{`s}
\DeclareMathSymbol{\tttau}{\mathord}{ttgreek}{`t}
\DeclareMathSymbol{\ttupsilon}{\mathord}{ttgreek}{`u}
\DeclareMathSymbol{\ttphi}{\mathord}{ttgreek}{`f}
\DeclareMathSymbol{\ttchi}{\mathord}{ttgreek}{`q}
\DeclareMathSymbol{\ttpsi}{\mathord}{ttgreek}{`y}
\DeclareMathSymbol{\ttomega}{\mathord}{ttgreek}{`w}

\DeclareMathSymbol{\ttAlpha}{\mathord}{ttgreek}{`A}
\DeclareMathSymbol{\ttBeta}{\mathord}{ttgreek}{`B}
\DeclareMathSymbol{\ttGamma}{\mathord}{ttgreek}{`G}
\DeclareMathSymbol{\ttDelta}{\mathord}{ttgreek}{`D}
\DeclareMathSymbol{\ttEpsilon}{\mathord}{ttgreek}{`E}
\DeclareMathSymbol{\ttZeta}{\mathord}{ttgreek}{`Z}
\DeclareMathSymbol{\ttEta}{\mathord}{ttgreek}{`H}
\DeclareMathSymbol{\ttTheta}{\mathord}{ttgreek}{`J}
\DeclareMathSymbol{\ttIota}{\mathord}{ttgreek}{`I}
\DeclareMathSymbol{\ttKappa}{\mathord}{ttgreek}{`K}
\DeclareMathSymbol{\ttLambda}{\mathord}{ttgreek}{`L}
\DeclareMathSymbol{\ttMu}{\mathord}{ttgreek}{`M}
\DeclareMathSymbol{\ttNu}{\mathord}{ttgreek}{`N}
\DeclareMathSymbol{\ttXi}{\mathord}{ttgreek}{`X}
\DeclareMathSymbol{\ttPi}{\mathord}{ttgreek}{`P}
\DeclareMathSymbol{\ttRho}{\mathord}{ttgreek}{`R}
\DeclareMathSymbol{\ttSigma}{\mathord}{ttgreek}{`S}
\DeclareMathSymbol{\ttTau}{\mathord}{ttgreek}{`T}
\DeclareMathSymbol{\ttUpsilon}{\mathord}{ttgreek}{`U}
\DeclareMathSymbol{\ttPhi}{\mathord}{ttgreek}{`F}
\DeclareMathSymbol{\ttChi}{\mathord}{ttgreek}{`Q}
\DeclareMathSymbol{\ttPsi}{\mathord}{ttgreek}{`Y}
\DeclareMathSymbol{\ttOmega}{\mathord}{ttgreek}{`W}

\newcommand{\sfS}{\mathsf{S}}

\newcommand{\bA}{\mathbf{A}}

\newcommand{\bc}{\mathbf{c}}

\newcommand{\bbE}{\mathbb{E}}

\newcommand{\bbN}{\mathbb{N}}

\newcommand{\bbP}{\mathbb{P}}

\newcommand{\bbR}{\mathbb{R}}

\DeclareSymbolFont{bsfletters}{OT1}{cmss}{bx}{n}
\DeclareSymbolFont{ssfletters}{OT1}{cmss}{m}{n}
\DeclareMathSymbol{\bsfGamma}{0}{bsfletters}{'000}
\DeclareMathSymbol{\ssfGamma}{0}{ssfletters}{'000}
\DeclareMathSymbol{\bsfDelta}{0}{bsfletters}{'001}
\DeclareMathSymbol{\ssfDelta}{0}{ssfletters}{'001}
\DeclareMathSymbol{\bsfTheta}{0}{bsfletters}{'002}
\DeclareMathSymbol{\ssfTheta}{0}{ssfletters}{'002}
\DeclareMathSymbol{\bsfLambda}{0}{bsfletters}{'003}
\DeclareMathSymbol{\ssfLambda}{0}{ssfletters}{'003}
\DeclareMathSymbol{\bsfXi}{0}{bsfletters}{'004}
\DeclareMathSymbol{\ssfXi}{0}{ssfletters}{'004}
\DeclareMathSymbol{\bsfPi}{0}{bsfletters}{'005}
\DeclareMathSymbol{\ssfPi}{0}{ssfletters}{'005}
\DeclareMathSymbol{\bsfSigma}{0}{bsfletters}{'006}
\DeclareMathSymbol{\ssfSigma}{0}{ssfletters}{'006}
\DeclareMathSymbol{\bsfUpsilon}{0}{bsfletters}{'007}
\DeclareMathSymbol{\ssfUpsilon}{0}{ssfletters}{'007}
\DeclareMathSymbol{\bsfPhi}{0}{bsfletters}{'010}
\DeclareMathSymbol{\ssfPhi}{0}{ssfletters}{'010}
\DeclareMathSymbol{\bsfPsi}{0}{bsfletters}{'011}
\DeclareMathSymbol{\ssfPsi}{0}{ssfletters}{'011}
\DeclareMathSymbol{\bsfOmega}{0}{bsfletters}{'012}
\DeclareMathSymbol{\ssfOmega}{0}{ssfletters}{'012}

\newcommand{\bphi}{\bm{\phi}}

\newcommand{\bXi}{\bm{\Xi}}

\makeatletter
\newcommand*\rel@kern[1]{\kern#1\dimexpr\macc@kerna}
\newcommand*\widebar[1]{%
  \begingroup
  \def\mathaccent##1##2{%
    \rel@kern{0.8}%
    \overline{\rel@kern{-0.8}\macc@nucleus\rel@kern{0.2}}%
    \rel@kern{-0.2}%
  }%
  \macc@depth\@ne
  \let\math@bgroup\@empty \let\math@egroup\macc@set@skewchar
  \mathsurround\z@ \frozen@everymath{\mathgroup\macc@group\relax}%
  \macc@set@skewchar\relax
  \let\mathaccentV\macc@nested@a
  \macc@nested@a\relax111{#1}%
  \endgroup
}
\makeatother

\DeclareMathOperator*{\argmax}{arg\,max}

\DeclareMathOperator{\ST}{s.t.\ }

\DeclareMathOperator{\var}{var}

\DeclareMathOperator{\cov}{cov}

\newcommand{\ifbcdot}[1]{\ifblank{#1}{\cdot}{#1}}

\DeclarePairedDelimiterX\abs[1]{\lvert}{\rvert}{\ifbcdot{#1}}
\DeclarePairedDelimiterX\parens[1]{(}{)}{\ifbcdot{#1}}
\DeclarePairedDelimiterX\brk[1]{[}{]}{\ifbcdot{#1}}
\DeclarePairedDelimiterX\braces[1]{\{}{\}}{\ifbcdot{#1}}
\DeclarePairedDelimiterX\angles[1]{\langle}{\rangle}{\ifblank{#1}{\cdot,\cdot}{#1}}
\DeclarePairedDelimiterX\ip[2]{\langle}{\rangle}{\ifbcdot{#1},\ifbcdot{#2}}
\DeclarePairedDelimiterX\norm[1]{\lVert}{\rVert}{\ifbcdot{#1}}
\DeclarePairedDelimiterX\ceil[1]{\lceil}{\rceil}{\ifbcdot{#1}}
\DeclarePairedDelimiterX\floor[1]{\lfloor}{\rfloor}{\ifbcdot{#1}}

\DeclareFontFamily{U}{matha}{\hyphenchar\font45}
\DeclareFontShape{U}{matha}{m}{n}{
      <5> <6> <7> <8> <9> <10> gen * matha
      <10.95> matha10 <12> <14.4> <17.28> <20.74> <24.88> matha12
      }{}
\DeclareSymbolFont{matha}{U}{matha}{m}{n}
\DeclareFontSubstitution{U}{matha}{m}{n}

\DeclareFontFamily{U}{mathx}{\hyphenchar\font45}
\DeclareFontShape{U}{mathx}{m}{n}{
      <5> <6> <7> <8> <9> <10>
      <10.95> <12> <14.4> <17.28> <20.74> <24.88>
      mathx10
      }{}
\DeclareSymbolFont{mathx}{U}{mathx}{m}{n}
\DeclareFontSubstitution{U}{mathx}{m}{n}

\DeclareMathDelimiter{\vvvert}{0}{matha}{"7E}{mathx}{"17}
\DeclarePairedDelimiterX\vertiii[1]{\vvvert}{\vvvert}{\ifbcdot{#1}}

\DeclarePairedDelimiterXPP\trace[1]{\operatorname{Tr}}{(}{)}{}{\ifbcdot{#1}} 
\DeclarePairedDelimiterXPP\col[1]{\operatorname{col}}{\{}{\}}{}{\ifbcdot{#1}} 
\DeclarePairedDelimiterXPP\row[1]{\operatorname{row}}{\{}{\}}{}{\ifbcdot{#1}} 
\DeclarePairedDelimiterXPP\erf[1]{\operatorname{erf}}{(}{)}{}{\ifbcdot{#1}}
\DeclarePairedDelimiterXPP\erfc[1]{\operatorname{erfc}}{(}{)}{}{\ifbcdot{#1}}
\DeclarePairedDelimiterXPP\KLD[2]{D}{(}{)}{}{\ifbcdot{#1}\, \delimsize\|\, \ifbcdot{#2}} 
\DeclarePairedDelimiterXPP\op[2]{\operatorname{#1}}{(}{)}{}{#2} 

\newcommand{\convp}{\stackrel{\mathrm{p}}{\longrightarrow}}

\newcommand{\T}{^{\mkern-1.5mu\mathop\intercal}}
\newcommand{\ud}{\,\mathrm{d}} 

\newcommand{\bzero}{\bm{0}}

\DeclarePairedDelimiterXPP\indicate[1]{{\bf 1}}{\{}{\}}{}{\ifbcdot{#1}}

\NewDocumentCommand\ofrac{s m}{%
	\IfBooleanTF#1%
	{\dfrac{1}{#2}}%
	{\frac{1}{#2}}%
}
\NewDocumentCommand\ddfrac{s m m}{%
	\IfBooleanTF#1%
	{\dfrac{\mathrm{d} {#2}}{\mathrm{d} {#3}}}%
	{\frac{\mathrm{d} {#2}}{\mathrm{d} {#3}}}%
}
\NewDocumentCommand\ppfrac{s m m}{%
	\IfBooleanTF#1%
	{\dfrac{\partial {#2}}{\partial {#3}}}%
	{\frac{\partial {#2}}{\partial {#3}}}%
}

\newcommand{\setgiven}{:}
\providecommand\given{}

\DeclarePairedDelimiterX\Set[2]\{\}{%
	\if#1:%
		\renewcommand\given{\SetSymbol{:}}%
	\else%
		\renewcommand\given{\SetSymbol[\delimsize]{#1}}%
	\fi%
#2
}

\NewDocumentCommand\set{s O{\setgiven} m}{%
	\IfBooleanTF#1%
	{\Set*{#2}{#3}}%
	{\Set{#2}{#3}}%
}

\NewDocumentCommand{\evalat}{ s O{\big} m e{_^} }{%
\IfBooleanTF{#1}%
{\left. #3 \right|}{#3#2|}%
\IfValueT{#4}{_{#4}}%
\IfValueT{#5}{^{#5}}%
}

\providecommand\given{}
\DeclarePairedDelimiterXPP\cprob[1]{}(){}{
\renewcommand\given{\nonscript\,\delimsize\vert\allowbreak\nonscript\,\mathopen{}}%
#1%
}
\DeclarePairedDelimiterXPP\cexp[1]{}[]{}{
\renewcommand\given{\nonscript\,\delimsize\vert\allowbreak\nonscript\,\mathopen{}}%
#1%
}

\DeclareDocumentCommand \P { s e{_^} d() g } {%
	\mathbb{P}%
	\IfBooleanTF{#1}%
		{
			\IfValueT{#2}{_{#2}}%
			\IfValueT{#3}{^{#3}}%
			\IfValueTF{#5}{\cprob{#4 \given #5}}{\IfValueT{#4}{\cprob{#4}}}%
		}%
		{
			\IfValueT{#2}{_{#2}}%
			\IfValueT{#3}{^{#3}}%
			\IfValueTF{#5}{\cprob*{#4 \given #5}}{\IfValueT{#4}{\cprob*{#4}}}%
		}%
}

\DeclareDocumentCommand \E { s e{_^} o g } {%
	\mathbb{E}%
	\IfBooleanTF{#1}%
		{
			\IfValueT{#2}{_{#2}}%
			\IfValueT{#3}{^{#3}}%
			\IfValueTF{#5}{\cexp{#4 \given #5}}{\IfValueT{#4}{\cexp{#4}}}%
		}%
		{
			\IfValueT{#2}{_{#2}}%
			\IfValueT{#3}{^{#3}}%
			\IfValueTF{#5}{\cexp*{#4 \given #5}}{\IfValueT{#4}{\cexp*{#4}}}%
		}%
}

\DeclareDocumentCommand \Var { s e{_^} d() g } {%
	\var%
	\IfBooleanTF{#1}%
		{
			\IfValueT{#2}{_{#2}}%
			\IfValueT{#3}{^{#3}}%
			\IfValueTF{#5}{\cprob{#4 \given #5}}{\IfValueT{#4}{\cprob{#4}}}%
		}%
		{
			\IfValueT{#2}{_{#2}}%
			\IfValueT{#3}{^{#3}}%
			\IfValueTF{#5}{\cprob*{#4 \given #5}}{\IfValueT{#4}{\cprob*{#4}}}%
		}%
}

\DeclareDocumentCommand \Cov { s e{_^} d() g } {%
	\cov%
	\IfBooleanTF{#1}%
		{
			\IfValueT{#2}{_{#2}}%
			\IfValueT{#3}{^{#3}}%
			\IfValueTF{#5}{\cprob{#4 \given #5}}{\IfValueT{#4}{\cprob{#4}}}%
		}%
		{
			\IfValueT{#2}{_{#2}}%
			\IfValueT{#3}{^{#3}}%
			\IfValueTF{#5}{\cprob*{#4 \given #5}}{\IfValueT{#4}{\cprob*{#4}}}%
		}%
}

\ExplSyntaxOn
\NewDocumentCommand \dist {m o o} {%
\mathrm{#1}\left(%
	\IfValueT{#3}{%
		\tl_if_blank:nTF{ #3 }{\cdot\, \middle|\, }{#3\, \middle|\, }%
	}
	\IfValueT{#2}{#2}%
\right)%
}
\ExplSyntaxOff

\NewDocumentCommand {\cbrace} {t+ D[]{black} D(){\widthof{#5}} m m } {%
	\begingroup%
		\color{#2}
		\IfBooleanTF{#1}{%
			\overbrace{#4}^%
		}{
			\underbrace{#4}_%
		}%
		{\parbox[c]{#3}{\centering\footnotesize{#5}}}%
	\endgroup%
}

\let\oldforall\forall
\renewcommand{\forall}{\oldforall \, }

\let\oldexist\exists
\renewcommand{\exists}{\oldexist \, }

\makeatletter

\newcommand{\rankcolor}[2]{%
	\expandafter\renewcommand\csname #1\endcsname[1]{%
		\ifblank{##1}{%
			{\color{#2} \textbf{#2}}%
		}{%
			\ifmmode
				\textcolor{#2}{\bm{##1}}%
			\else%
				{\color{#2} \textbf{##1}}%
			\fi	
		}%
	}
}

\rankcolor{first}{red}
\rankcolor{second}{blue}
\rankcolor{third}{cyan}
\makeatother

\graphicspath{{./Figures/}{./figures/}}
\DeclareDocumentCommand{\includeCroppedPdf}{ o O{./Figures/} m }{
	\IfFileExists{#2#3-crop.pdf}{}{%
		\immediate\write18{pdfcrop #2#3.pdf #2#3-crop.pdf}}%
	\includegraphics[#1]{#2#3-crop.pdf}
}

\makeatletter
\newcommand*{\addFileDependency}[1]{
  \typeout{(#1)}
  \@addtofilelist{#1}
  \IfFileExists{#1}{}{\typeout{No file #1.}}
}
\makeatother

\definecolor{gray90}{gray}{0.9}
\def\colorlist{red,blue,brown,cyan,darkgray,gray,lightgray,green,lime,magenta,olive,orange,pink,purple,teal,violet,white,yellow}

\makeatletter
\def\startcomment{[}
\ifx\nohighlights\undefined
	\newcommand{\createcolor}[1]{%
			\expandafter\newcommand\csname #1\endcsname[1]{{\color{#1} ##1}}%
	}
	\newcommand{\msout}[1]{\text{\color{green} \st{\ensuremath{#1}}}}
	\newcommand{\del}[1]{{\color{green}\ifmmode \msout{#1}\else\st{#1}\fi}}
\else
	\newcommand{\createcolor}[1]{%
			\expandafter\newcommand\csname #1\endcsname[1]{%
				\noexpandarg%
				\StrChar{##1}{1}[\firstletter]%
				\if\firstletter\startcomment%
					\relax
				\else%
					##1
				\fi
			}%
	}
	\newcommand{\msout}[1]{}
	\newcommand{\del}[1]{}
\fi

\def\@tempa#1,{%
    \ifx\relax#1\relax\else
        \createcolor{#1}%
        \expandafter\@tempa
    \fi
}
\expandafter\@tempa\colorlist,\relax,
\makeatother

\newcommand{\hhide}[1]{}

\ifx\diagnoselabel\undefined
	\relax
\else
	\makeatletter
	\def\@testdef #1#2#3{%
		\def\reserved@a{#3}\expandafter \ifx \csname #1@#2\endcsname
			\reserved@a  \else
			\typeout{^^Jlabel #2 changed:^^J%
				\meaning\reserved@a^^J%
				\expandafter\meaning\csname #1@#2\endcsname^^J}%
			\@tempswatrue \fi}
	\makeatother
\fi

\newacronym{GGSP}{GGSP}{generalized graph signal processing}
\newacronym{GSO}{GSO}{graph shift operator}
\newacronym{GSP}{GSP}{graph signal processing}
\newacronym{MHT}{MHT}{multiple hypothesis testing}
\newacronym{BFDR}{BFDR}{Bayes false discovery rate}
\newacronym{FDR}{FDR}{false discovery rate}
\newacronym{FWER}{FWER}{family-wise error rate}
\newacronym{TPR}{TPR}{true positive rate}
\newacronym{lfdr}{lfdr}{local false discovery rate}
\newacronym{pdf}{PDF}{probability density function}
\newacronym{cdf}{CDF}{cumulative distribution function}
\newacronym{KKT}{KKT}{Karush–Kuhn–Tucker}
\newacronym{WLLN}{WLLN}{weak law of large numbers}
\newacronym{CSLLN}{CSLLN}{conditional strong law of large numbers}
\newacronym{DCT}{DCT}{dominated convergence theorem}
\newacronym{BIC}{BIC}{Bayesian information criterion}
\newacronym{MLE}{MLE}{maximum likelihood estimator}
\newacronym{a.e.}{a.e.}{almost everywhere}
\newacronym{a.s.}{a.s.}{almost surely}
\newacronym{r.c.d.}{r.c.d.}{regular conditional distribution}
\newacronym{AWGN}{AWGN}{additive white Gaussian noise}
\newacronym{BH}{BH}{Benjamini-Hochberg}
\newacronym{GLM}{GLM}{generalized linear model}
\newacronym{ADMM}{ADMM}{alternating direction method of multipliers}
\newacronym{LOS}{LOS}{line-of-sight}

\newcommand{\FDR}{\mathrm{FDR}}
\newcommand{\mFDR}{\mathrm{mFDR}}
\newcommand{\pow}{\mathrm{pow}}
\newcommand{\mpow}{\mathrm{mpow}}
\newcommand{\lfdr}{\mathrm{lfdr}}

\newcommand{\Ind}{\mathbb I}

\DeclareMathOperator*\lowlim{\underline{lim}}
\DeclareMathOperator*\uplim{\overline{lim}}

\newcommand{\altReg}{\mathcal{H}_1}

\newcommand{\altRegEst}{\widehat{\mathcal{H}}_1}

\newcommand{\bfdr}[1][\rejRegP]{\mathrm{BFDR}\!\left(#1\right)}
\renewcommand{\convp}{\xlongrightarrow{\mathsf{p}}}

\renewcommand{\emph}[1]{\textit{#1}}
\newcommand{\epoch}{t}
\newcommand{\epochRV}{\mathrm{T}}

\newcommand{\epochDom}{\calT}

\newcommand{\FP}[1][p]{F_{\mathrm{P}}\!\left(#1\right)}
\newcommand{\fP}[1][p]{f_{\mathrm{P}}\!\left(#1\right)}
\newcommand{\FPNodeEpoch}[1][p]{F_{\pRV}\!\left(#1\right)}
\newcommand{\fPNodeEpoch}[1][p]{f_{\pRV}\!\left(#1\right)}

\newcommand{\fPAlt}[1][p]{f_{\mathrm{P}|\HAlt\!}\!\left(#1\right)}
\newcommand{\FPAltNodeEpoch}[1][p]{F_{\pRV|\HAlt\!}\!\left(#1\right)}
\NewDocumentCommand{\fPAltNodeEpoch}{O{p} O{v} O{t}}{f_{\mathrm{P}_{(#2,#3)}|\HAlt\!}\!\left(#1\right)}
\newcommand{\FPNul}[1][p]{F_{\mathrm{P}|\HNul\!}\!\left(#1\right)}
\newcommand{\FpNulNodeEpoch}[1][p]{F_{\pRV|\HNul\!}\!\left(#1\right)}
\newcommand{\fPNul}[1][p]{f_{\mathrm{P}|\HNul\!}\!\left(#1\right)}
\NewDocumentCommand{\fPNulNodeEpoch}{O{p} O{v} O{t}}{f_{\mathrm{P}_{(#2,#3)}|\HNul\!}\!\left(#1\right)}

\newcommand{\graph}{\mathcal{G}}

\newcommand{\HAlt}{\mathsf{H}_1}
\newcommand{\HNul}{\mathsf{H}_0}

\newcommand{\HTrue}{\mathrm{H}_{\tuple}}
\newcommand{\HTrueRV}{\mathrm{H}_{(\nodeRV,\epochRV)}}
\newcommand{\HEst}{\widehat{\mathrm{H}}_{\tuple}}

\newcommand{\idc}[1]{\Ind\set{#1}}

\newcommand{\jointDom}{\vertexDom\times\epochDom}

\newcommand{\node}{v}
\newcommand{\nodeRV}{\mathrm{V}}

\newcommand{\nulReg}[1][t]{\mathcal{H}_0}
\newcommand{\nulRegEst}[1][t]{\widehat{\mathcal{H}}_0}

\newcommand{\numTests}{I}

\newcommand{\p}{\rVar{P}_{\tuple}}
\newcommand{\pRV}{\mathrm{P}_{\tuple}}

\newcommand{\pRVNodeEpochRV}{\mathrm{P}_{(\nodeRV,\epochRV)}}
\newcommand{\pVec}{\rVec{p}}
\newcommand{\pSet}[1][\samReg]{\mathcal{P}_{#1}}

\newcommand{\rejRegP}[1][(h)]{\mathcal{R}_{\mathrm{P}}{#1}}
\newcommand{\rVar}[1]{{\mathrm{#1}}}
\newcommand{\rVec}[1]{{\boldsymbol{\mathrm{#1}}}}

\newcommand{\samReg}{\mathcal{H}}
\newcommand{\signal}{\upgamma}
\newcommand{\signalNodeEpoch}{\signal_{\tuple}}
\newcommand{\tuple}{(\node, \epoch)}
\newcommand{\tupleRV}{(\nodeRV, \epochRV)}

\newcommand{\vertexDom}{\calV}

\let\isout\sout \renewcommand{\sout}[1]{\ifmmode\text{\isout{\ensuremath{#1}}}\else\isout{#1}\fi}

\graphicspath{{./Figures/}}

\begin{document}

\title{A Graph Signal Processing Perspective of Network Multiple Hypothesis Testing with False Discovery Rate Control}
\author{Xingchao Jian, Martin Gölz, Feng Ji, Wee~Peng~Tay and Abdelhak M. Zoubir
\thanks{
This work was supported by the xxxxx. 
}%
}



\maketitle \thispagestyle{empty}


\begin{abstract}
We consider a multiple hypothesis testing problem in a sensor network over the joint spatio-temporal domain. The sensor network is modeled as a graph, with each vertex representing a sensor and a signal over time associated with each vertex. We assume a hypothesis test and an associated $p$-value for every sample point in the joint spatio-temporal domain. Our goal is to determine which points have true alternative. By parameterizing the unknown $p$-value distribution under the alternative and the prior probabilities of hypotheses being null with a bandlimited generalized graph signal, we can obtain consistent estimates for them. Consequently, we also obtain an estimate of the local false discovery rates (lfdr). We prove that by using a step-up procedure on the estimated lfdr, we can achieve asymptotic false discovery rate control at a pre-determined level. Numerical experiments validate the effectiveness of our approach compared to existing methods.
\end{abstract}


\section{Introduction}\label{sect:intro}

In modern data analysis, graphs are commonly used to represent entities and relationships, making graph-based problems prevalent. \Gls{GSP} has emerged as a powerful tool to address these problems by leveraging the domain knowledge of graphs \cite{ShuNarFroOrt:J13,LeuMarMou:J23, JianJiTay:J23}. It extends the traditional signal processing domain from discrete-time stamps to vertices in a network \cite{SanMou:J14}. With techniques such as filtering \cite{PavGirOrt:J23, IsuGamShuSeg:J24}, sampling \cite{AniGadOrt:J16, TanEldOrt:J20}, representation \cite{ShuRicVan:J16}, and reconstruction \cite{RomMaGia:J16,KroRouEld:J22}, \gls{GSP} has found applications in various domains, including point cloud denoising \cite{ZengCheNg:J19, DinCheBaj:J20,HuPangLiu:J21}, image processing \cite{YagOzg:J20,YangXuHou:J23}, and brain network analysis \cite{HuangBolMed:J18}.

The traditional \gls{GSP} literature primarily investigates signal relationships through graph structures. In practice, additional side information, such as temporal relationships, may also be included. To accommodate both the discrete-time domain and the graph domain, \cite{GraLouPerRic:J18} adopts a Cartesian product graph. This approach allows discrete-time signals defined on the graph to be regarded as classical graph signals on a product graph, enabling the application of signal processing techniques to this product graph \cite{PerLouGraVan:C17,LouPer:J19}.

This framework is further extended in \cite{JiTay:J19} to include a more general Hilbert space as the side-information domain, allowing for asynchronous sampling on vertices and better exploitation of domain knowledge, especially in the case of a continuous domain \cite{JianTayEld:J23}. This extended framework, referred to as \gls{GGSP}, defines signals on a joint domain formed by the graph domain and the side-information domain, rather than solely on a graph. In this paper, we consider problems under this \gls{GGSP} framework for greater generality and flexibility. For simplicity, \emph{we often refer to the side information domain colloquially as the ``time domain'' and its samples as ``(time) instances''}, although our results apply to a general side-information domain that may not necessarily be interpreted as time.

By leveraging the graph topology, \gls{GSP} estimation methods typically associate a noisy and incomplete input with a smooth approximation of the target graph signals \cite{RomMaGia:J16,KroRouEld:J22}. However, there are only a few existing works that address detection problems on graphs. In our work, we have a hypothesis testing problem for every pair of vertex and instance. Given the availability of $p$-values for all such pairs, our goal is to determine a subset in which the alternatives are true. Similar to estimation theory, we expect that by leveraging the joint domain information, we can map the set of $p$-values to a subset of the joint domain that identifies the set where the alternatives are true. Furthermore, this identified set should exhibit smoothness in the joint domain.

In statistics, the problem of determining which hypotheses to reject among a large set of hypotheses based on their associated $p$-values is known as the \gls{MHT} problem \cite{Efr:10}. The seminal work by \cite{BenHoc:J95} proposes an adaptive thresholding method, also known as the \gls{BH} method, which controls the \gls{FDR} to ensure that the identified set of alternatives does not contain too many misidentified hypotheses. This method is improved by \cite{Sto:J02} by estimating the proportion of null hypotheses. In recent years, there have been more \gls{MHT} works that incorporate structural information. For example, \cite{LiBar:J18} proposes a weighted \gls{BH} method that assigns weights to $p$-values based on prior knowledge of structural information, such as low variation over a graph. The resulting \gls{FDR} is shown to be related to the Rademacher complexity of the feasible set of weights.

The paper \cite{TanKoyPolSco:J18} proposes a method for solving the \gls{MHT} problem over a graph. This method leverages the graph information by encoding the prior probabilities of being null as a function of a low variation signal on the graph. 
The paper \cite{LeiFit:J18} develops a method to iteratively threshold, mask, and reveal the $p$-values to control the \gls{FDR} at a pre-determined level. This approach allows for the incorporation of prior knowledge during the thresholding step. Furthermore, the authors of \cite{LeiFit:J18} prove that the optimal rejection strategy, in terms of both marginal power and marginal \gls{FDR}, is to set a threshold on the \gls{lfdr} values. 
Another work, \cite{CaoChenZhang:J22}, demonstrates that this strategy can asymptotically control the \gls{FDR} at a pre-determined level under certain consistency conditions. The authors also propose an EM algorithm for estimating the unknown \glspl{pdf} of $p$-values and the prior probabilities of being null. 
The work \cite{GolZouKoi:J22} proposes a method of estimating a Beta mixture density, and applies it to the detection problem of spatial signals. 
The work \cite{PouXia:J23} proposes distributed methods for approximating the \gls{BH} method and scan procedure \cite{AriChenYing:J21} over graphs. This method incorporates the case where the prior probabilities of being null vary over the graph.

In this paper, we propose an approach for simultaneously estimating the prior probability of hypotheses being null and the distribution of $p$-values under alternative in a sensor network over the joint spatio-temporal domain.
Unlike existing methods \cite{GolZouKoi:J22,LiBar:J18,TanKoyPolSco:J18,PouXia:J23} that assume these quantities to be constant, piecewise constant, or with finite choices, we parameterize them using an underlying generalized graph signal (cf.\ \cref{sect:prelim:GGSP}), resulting in a finer model. We model the underlying generalized graph signal as a bandlimited signal defined by a finite number of parameters.
This strategy allows us to utilize the \gls{MLE}, which is known to be consistent. We prove the asymptotic control of the \gls{FDR} by applying the \gls{lfdr} based approach. Our method is more suitable for real-world problems, computationally less challenging, and not restricted by uncontrollable parameters, making it an improvement over existing approaches.

Our main contributions are summarized as follows:
\begin{enumerate}
\item We propose a method for simultaneous estimation of the \glspl{pdf} of $p$-values that are inhomogeneous over the joint domain. In \cref{thm:unif_conv_para}, we prove that by utilizing the domain knowledge, this method yields consistent estimates.
\item We provide an \gls{MHT} strategy with simultaneously varying null proportion and $p$-value distribution under the alternative under the \gls{GGSP} framework. In \cref{thm:limit_FDR}, we provide conditions for asymptotic \gls{FDR} control and derive the asymptotic power of the proposed method in \cref{thm:limit_pow}.
\item We validate the utility of the proposed strategy in practical applications.
\end{enumerate}

The remainder of this paper is organized as follows. In \cref{sect:stat_model}, we introduce the empirical-Bayesian model for deriving our proposed method. In \cref{sect:MHT-GGSP}, we derive the method for \glspl{pdf} estimation from $p$-values. We prove that using the \gls{lfdr} approach with the proposed \gls{pdf} estimate, we obtain asymptotic \gls{FDR} control. In \cref{sect:exp}, we validate our proposed method on communication and seismic datasets. We conclude in \cref{sect:conc}.

\emph{Notation:} 
Upright letters like \(\rVar{A}\) and upright bold lowercase letters like \(\rVec{a}\) denote random variables and random vectors with corresponding \glspl{pdf} \(f_\rVar{A}(a)\) (univariate) and \(f_\rVec{a}(\boldsymbol{a})\) (multivariate), respectively. Upright bold letters like \(\rVar{\bA}\) represent random matrices. Deterministic scalars, vectors and matrices are denoted by italic letters \(A\), italic bold lowercase letters \(\boldsymbol{a}\) and italic bold uppercase letters \(\boldsymbol{A}\), respectively. A calligraphic letter \(\mathcal{A}\) represents a set with cardinality \(|\mathcal{A}|\).
We use $\uplim$ and $\lowlim$ to denote the upper and lower limits. We write $\convp$ for convergence in probability. The indicator function \(\idc{\text{cond}}\) returns \(1\) if cond is true and \(0\) otherwise.
Let $1 \vee A = \max(1, A)$.
The abbreviation ``w.p.'' stands for ``with probability''. For two functions $f$ and $g$, we write $f\circ g$ to represent their composition.
A glossary of frequently used symbols is provided in \cref{tab:notation}. 

\begin{table}[!htb]
    \centering
    \caption{List of frequently used symbols}
    \label{tab:notation}
    \begin{tabular}{ >{\raggedright}p{5cm} p{8cm} }
        \toprule
        \textbf{Symbol} & \textbf{Description} \\
        \midrule
        $\graph = (\calV,\calE)$ & Graph \(\graph\) with vertex set $\calV$ and edge set $\calE$, and $|\calV|=N$ \\
        $\boldsymbol{S}$ & Graph shift operator \\
        $\set{\lambda_k\given k=1,\dots, N}$ & Graph frequencies \\
        $\set{\bphi_k\given k=1,\dots, N}$ & Graph Fourier basis \\
        $\set{\psi_k\given k=1,2,\dots}$ & Fourier basis of $L^2(\calT)$\\
        $\HNul,\, \HAlt$ & Null and alternative hypotheses, respectively \\
        $\HTrue$ & Ground-truth hypothesis at \(\tuple\in\samReg\) \\
        \(\samReg\) & Discrete subset of \(\jointDom\) on which inference is conducted\\        
        $\nulReg\subset\samReg$, $\altReg=\samReg\setminus\nulReg$ & Null and alternative regions, respectively \\
        $h$ & Detection strategy \\
        \(\nulRegEst\), \(\altRegEst\) & Estimated null and alternative regions, respectively \\
        $\FDR_{\samReg}(h)$ & \gls{FDR} of $h$, defined by \cref{eq:def_fdr}\\
        $\pow_{\samReg}(h)$ & power of $h$, defined by \cref{eq:def_pow} \\
        $\alpha$ & Nominal \gls{FDR} level \\
        $\signalNodeEpoch\equiv\signal\tuple$ & Stochastic process on $\calV\times\calT$\\
        $\bXi$ & Fourier coefficients of $\signal$, see \cref{asp:bandlimit}\\
        \(\pi_0\circ\signalNodeEpoch\equiv\pi_0\circ\signal\tuple\) & Prior probability of the null hypothesis at $(v,t)$\\
        $\fPNulNodeEpoch$,$\FpNulNodeEpoch$ & $p$-value \gls{pdf}, \gls{cdf} under \(\HNul\)\\
        $\fPAltNodeEpoch$, \(\FPAltNodeEpoch\) & $p$-value \gls{pdf}, \gls{cdf} under \(\HAlt\) \\
        $\fPNodeEpoch, \FPNodeEpoch$ & Marginal $p$-value \gls{pdf}, \gls{cdf}\\
        \bottomrule
    \end{tabular}
    
\end{table}

\section{Statistical Model}\label{sect:stat_model}

We consider a graph $\graph=(\calV,\calE)$, where $\calV$ represents the vertex set and $\calE$ is the edge set. We assume that \(\graph\) is finite, undirected, and without self-loops. A graph signal is defined as a function from $\calV$ to $\Real$. Since $\abs{\calV}=N$ is finite, a graph signal can also be written as an $N$-dimensional vector. 


In this paper, apart from the graph \(\graph\), we consider an additional space $\calT$ as the side-information space. We assume that $\calT$ is a topological measure space such that every open set is measurable. This measure space represents the domain of the signal observed by each vertex, i.e., we consider spatio-temporal signals $g : \calV\times\calT \to \bbR$. An example of $\calT$ is a finite time interval $[a,b]$. This allows us to model the \gls{MHT} problem on the joint domain $\calV\times\calT$ as determining the state of the spatial-temporal graph signal on a discrete subset of $\jointDom$. For the sake of readability, a quantity \(g(v, t)\) that depends on \(\tuple \in \jointDom\) is also denoted as \(g_{\tuple}\). Throughout this paper, we assume that functions of random elements are always measurable or have measurable choices.

\subsection{Problem formulation}

The spatial inference objective is the identification of the true state \(\HTrue \in \{\HNul, \HAlt\} \) of the monitored underlying spatio-temporal signal for each combination of node and instance \(\tuple \in \jointDom\). 
If the signal does not correspond to the detection objective, the \textit{null hypothesis} \(\HNul\) holds. Otherwise, the \textit{alternative} \(\HAlt\) is in place. For example, in a communication sensor network, if the objective is to detect which sensors receive informative signals, sensors receiving signals above a certain power level are assigned $\HAlt$ and assigned $\HNul$ otherwise. 

In practice, the joint domain \(\jointDom\) is sampled at a finite number of vertex-instance tuples \(\tuple\in\samReg \subset\jointDom\) of size \(|\samReg| = \numTests\). For each \(\tuple\in\samReg\), a single local summary statistic is observed, quantifying the belief in the local state of the monitored spatio-temporal signal. In line with previous work on spatial inference via \gls{MHT} \cite{Goelz2022b, Goelz2022c, Goelz2023a, Goelz2024b}, we use \(p\)-values as local summary statistics. A small \(p\)-value indicates little support for the null hypothesis. The set of all available \(p\)-values is \(\pSet = \set*{\p \given \tuple\in\samReg}\) with \(|\pSet| = |\samReg| = \numTests\). We write \(\pVec\in[0, 1]^{\numTests}\) for a vector containing the elements of \(\pSet\). Solving the spatial inference problem is equivalent to finding two mutually exclusive subsets of \(\samReg\), the \emph{null region} \(\nulReg = \set*{\tuple \in \samReg \given \HTrue = \HNul}\) and \emph{alternative region} \(\altReg = \set*{\tuple \in \samReg \given \HTrue = \HAlt}\). A \textit{detection strategy} is a mapping
\begin{align*}
h:[0, 1]^{\numTests}&\to\set{\HNul,\HAlt}^{\samReg},\\
\pVec&\mapsto h(\pVec),
\end{align*}
i.e., given a vector of $p$-values $\pVec$, the function $h(\pVec)$ makes the spatio-temporal decision \(\HEst\in\{\HNul, \HAlt\}\) for all \(\tuple\in\samReg\). In other words, \(h\) finds the estimated null and alternative regions \(\nulRegEst=\{\tuple\in\samReg: \HEst = \HNul\}\) and \(\altRegEst=\{\tuple\in\samReg: \HEst = \HAlt\}\), respectively. Any valid detection strategy provides statistical guarantees with respect to a specified error measure. In traditional binary testing, the objective is typically to control the type I error, i.e., the probability of deciding in favor of \(\HAlt\) when \(\HNul\) is true, at a pre-specified level. When comparing two detection strategies that require similar assumptions, the better strategy provides (at the same nominal false alarm level) a higher detection power, i.e., probability of deciding in favor of \(\HAlt\) when it is true. 

In \gls{MHT}, multiple statistical measures to quantify false alarms exist. The most popular are the \gls{FWER} \cite{Tukey1991} and the \gls{FDR} \cite{BenHoc:J95}. The \gls{FWER} is the probability of at least one false positive, i.e., the decision in favor of \(\HAlt\) when \(\HNul\) is true, while the \gls{FDR} is the expected proportion of false positives among all positives. In general, detection strategies designed to control the \gls{FDR} yield more positives than \gls{FWER} controlling strategies, because a higher absolute number of false positives is allowed, as long as also more true positives are obtained. We are interested in finding the null and alternative regions with high accuracy in spatial inference. Hence, we consider \gls{FDR} controlled strategies \cite{GolZouKoi:J22}.

Formally, the \emph{\gls{FDR}} and the \emph{power} (which we also refer to as \emph{\gls{TPR}} in the \gls{MHT} context) achieved by $h$ are defined as 
\begin{align}
    \FDR_{\samReg}(h) &= \bbE[\frac{\abs{\altRegEst\cap \nulReg}}{1\vee |\altRegEst|}] = \bbE[\frac{\sum_{\tuple \in \altRegEst}\idc{\tuple \in\nulReg}}{1\vee \sum_{\tuple\in\samReg}\idc{\tuple\in\altRegEst}}] \label{eq:def_fdr},\\
    \pow_{\samReg}(h) &=\bbE[\frac{|\altRegEst\cap \altReg|}{1\vee |\altReg|}] = \bbE[\frac{\sum_{\tuple\in\altRegEst}\idc{\tuple\in\altReg}}{1\vee\sum_{\tuple\in\samReg}\idc{\tuple\in\altReg}}]\label{eq:def_pow}.
\end{align}
Our goal is to design a strategy $h$ such that $\pow_{\samReg}(h)$ is as high as possible while $\FDR_{\samReg}(h)$ is controlled at a pre-defined level, i.e., \(\FDR_{\samReg}(h)\leq\alpha\). 

In a series of works \cite{Efron2001, Efron2005, Efron2008}, an empirical Bayes framework was developed for controlling the \gls{FDR} in large-scale testing situations. In contrast to traditional single binary testing, in \gls{MHT}, the statistical models and priors can be learned from the available set of test statistics. This philosophy has been applied successfully for spatial inference in large-scale sensor networks \cite{GolZouKoi:J22, Goelz2022b, Goelz2024b}. Also, additional available \textit{contextual information} can be incorporated into the probability models to enhance detection power \cite{TanKoyPolSco:J18, Goelz2022c}. In this work, the contextual information comes in the form of the joint domain $\jointDom$ which encodes the spatio-temporal structure of the observed signal. In what follows, we first introduce the Bayesian approach to \gls{MHT} before formally proposing the incorporation of the graph structure.  

In empirical Bayes inference, the marginal \gls{pdf} of the test statistics in \(\pSet\),
\begin{equation}
\label{eq:two-groups-pdf}
    f_\mathrm{P}(p) = \pi_0 \cdot \fPNul + (1-\pi_0)\cdot \fPAlt,
\end{equation}
is expressed as a mixture of the distribution under the null \(\fPNul\) and the distribution under the alternative \(\fPAlt\). In this \textit{two-groups model} \cite{Efron2008}, \(\fPNul\) \big(\(\fPAlt\)\big) captures the statistical behavior of the \(p\)-values from all \(\tuple\in\nulReg\) \big(\(\tuple\in\altReg\)\big), i.e., wherever the null (alternative) hypothesis is in place. The ground truth hypotheses are regarded as random by this model, and $\pi_0$ is the prior probability that the null hypothesis is true for each test. With the two-groups model, the posterior probability of \(\HNul\) given the local summary statistic \(\p\) is the so-called \textit{\gls{lfdr}} \cite{Efr:10} obtained by applying Bayes theorem: 
\begin{align}\label{eq:lfdrnonlocal}
    \lfdr(p) = \bbP(\HTrue= \HNul\,|\,\p=p) = \frac{\pi_0\cdot \fPNul[p]}{\fP[p]}.
\end{align}
While \glspl{lfdr} themselves are informative as they express a belief in the state of the observed signal, they can also be used to find \(\altRegEst\) under \gls{FDR} control at nominal level \(\alpha\). To this end, denote by \(\rejRegP\subset[0, 1]\) the \textit{rejection region} of a detection strategy  \(h\). If \(\p\in \rejRegP\), we declare \(\HEst=\HAlt\), i.e., that the alternative is in place at the associated \(\tuple\in\samReg\). The expected \gls{lfdr} across the rejection region is 
\begin{equation}
    \label{eq:bfdr}
    \bbE[\lfdr(\p)\given\p\in\rejRegP] = \frac{\pi_0\cdot \FPNul[\rejRegP]}{\FP[\rejRegP]} \equiv \bfdr
\end{equation}
the so-called \textit{\gls{BFDR}}. Here, \(\FP[\rejRegP] = \int_{\rejRegP}\pi_0\cdot \fPNul \mathrm{d}p+ \int_{\rejRegP}(1-\pi_0)\cdot \fPAlt\mathrm{d}p\) and \(\FPNul[\rejRegP] = \int_{\rejRegP}\fPNul\mathrm{d}p\) are the total mass and mass under \(\HNul\) of \(\mathrm{P}\) concentrated in the rejection region \(\rejRegP\), respectively. The \gls{BFDR} dominates the \gls{FDR} from \cref{eq:def_fdr} \cite{Efr:10}. Consequently, any detection strategy \(h\) that selects \(\rejRegP\) such that \(\bfdr\leq \alpha\) controls the \gls{FDR} at nominal level \(\alpha\). In order to make as many discoveries as possible with \gls{FDR} control, it has been proposed to choose the largest set of hypotheses whose average \gls{lfdr} value does not exceed the nominal level \cite{Halme2019}. 

The two-groups model implicitly assumes that the tests are exchangeable, since \(\fPNul, \fPAlt\) and \(\pi_0\) are the same for all hypotheses. However, this exchangeability assumption prevents profiting from more localized models, when contextual information is available. For example, in \cite{TanKoyPolSco:J18}, the authors considered the situation where each hypothesis is associated with a vertex in a graph. They proposed to allow $\pi_0$ to vary across vertices to raise the detection power. The idea of a spatially varying prior of the null hypothesis has also been considered in \cite{Goelz2022c}, without using a graph structure. 

\subsection{A \gls{GGSP} model for \gls{MHT}}

In this paper, we allow both the prior and the probability distributions to vary across space and time. This is achieved by deploying the hierarchical model depicted in \cref{fig:Bayes:model}.

\begin{figure}[!htb]
\centering
\includegraphics[scale=1.3]{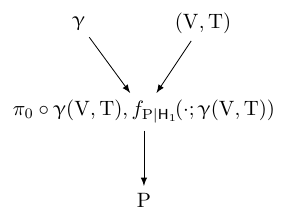}
\caption{The scheme of the hierarchical Bayesian model for \gls{MHT} in \gls{GGSP}.}
\label{fig:Bayes:model}
\end{figure}

In this model, $\signal$ is a stochastic process on \(\jointDom\) with $\signalNodeEpoch \equiv \signal\tuple \in \calZ$ for all $(v,t)\in\jointDom$, where $\calZ$ is a given subset of a Euclidean space. We use $\signal$ to characterize (i.e., parameterize) the distributions of $p$-values over different points in $\jointDom$. Note that $\signal$ is different from the monitored underlying spatio-temporal signal whose state at $(v,t)$ is $\HTrue \in \set{\HNul, \HAlt}$. As illustrated in \cref{exam:moving_trans}, $\signal$ models the randomness of the monitored signal strength at different sample points.

As is often the case, signals corresponding to $\HAlt$ often occur in clusters \cite{TanKoyPolSco:J18,GolZouKoi:J22}. Therefore, in this paper, we assume that $\signal$ has a low total variation over $\jointDom$. To formally incorporate this assumption, we model it as a random \emph{generalized graph signal} in the context of \gls{GGSP} \cite{JiTay:J19,JianTayEld:J23}. The generalized graph signal space is the function space $L^2(\jointDom)$. The shift operator on $L^2(\calV)$ is defined as a symmetric matrix $\boldsymbol{S}=(s_{uv})\in\Real^{N\times N}$ whose non-zero entries reflect the graph structure, i.e., $s_{uv}\neq0$ if and only if $(u,v)\in\calE$. We refer to $\boldsymbol{S}$ as a \emph{\gls{GSO}}. Typical choices of the \gls{GSO} include the graph adjacency matrix or the Laplacian matrix. 

Suppose $\boldsymbol{S}$ has orthonormal eigenvectors $\set{\bphi_1,\dots,\bphi_N}$ with corresponding eigenvalues $\set{\lambda_1,\allowbreak\dots,\allowbreak\lambda_N}$. These eigenvalues are referred to as graph frequencies since each $\lambda_k$ reflects the smoothness (i.e., total variation) of the signal $\bphi_k$ on \(\graph\) \cite{ShuNarFro:J13,SanMou:J14,JianJiTay:J23}. The notion of total variation varies with the choice of $\boldsymbol{S}$. The set of vectors $\set{\bphi_1,\dots,\bphi_N}$ is called the \emph{graph Fourier basis}. In this paper, we use the convention that $\lambda_k$ are listed in increasing order of the total variation of $\bphi_k$, that is, $\bphi_1$ and  $\bphi_N$ are the smoothest and least smooth graph signals among all graph Fourier bases, respectively. 

The shift operator on $\calT$ is defined as a compact, self-adjoint operator on $L^2(\calT)$, with its eigenvectors $\set{\psi_k\given k\geq1}$ forming an orthonormal basis for $L^2(\calT)$. Consequently, we obtain a basis for $L^2(\jointDom)$ through the tensor product: $\set{\bphi_{k_1}(v)\psi_{k_2}(t)\given k_1=1,\dots,N, k_2\geq1}$. This basis represents orthogonal signals of varying smoothness over $\jointDom$ (cf.\ \cref{sect:prelim:GGSP}). In \cref{subsect:inference}, we enforce smoothness on $\signal$ by requiring it to be bandlimited.

Independently of the process $\signal$, \(\numTests\) sample points $\tuple$ are drawn from a strictly positive probability measure $\rho$ on $\jointDom$ and form the set of vertex-instance tuples \(\samReg\). By assuming $\rho$ as strictly positive, we require that any non-empty open set has a positive measure under $\rho$. Hence, the \(\tuple\in\samReg\) are observations of the random element \(\tupleRV\in\jointDom\).

Let the \textit{null proportion} $\pi_0$ be a continuous function from $\calZ$ to $[0,1]$. Given $\signal$ and $\tuple$, we obtain a probability $\pi_0\circ\signalNodeEpoch := \pi_0(\signalNodeEpoch) \in[0,1]$.
With this spatio-temporally variable model of the prior null proportion, the true hypotheses for each \(\tuple\in\samReg\) are generated by
\begin{align*}
\HTrue=
\begin{cases}
\HNul & \mathrm{w.p.~} \pi_0\circ\signalNodeEpoch, \\
\HAlt & \mathrm{w.p.~} 1-\pi_0\circ\signalNodeEpoch.
\end{cases}
\end{align*}
Finally, the \(p\)-values are a function \(\pRV\) of \(\tuple\) and follow different models under the null and alternative hypotheses:
\begin{align*}
    \pRV \sim \begin{cases}
        \fPNulNodeEpoch& \qquad \HTrue= \HNul, \\
        \fPAltNodeEpoch \equiv f_{\mathrm{P}|\HAlt}\big(p;\signalNodeEpoch\big)  & \qquad \HTrue= \HAlt.
    \end{cases}
\end{align*}
Under \(\HNul\), the \(p\)-value on $\tuple$ follow the distribution $\fPNulNodeEpoch$. This distribution does not vary with $\tuple$ in most practical applications. Under the alternative, the distribution depends on \(\signalNodeEpoch\). This is in line with the traditional two-groups model, because appropriately calculated \(p\)-values are uniformly distributed on \([0, 1]\) under \(\HNul\) and their distribution under \(\HAlt\) depends on the size of the effect, i.e., the strength of the signal that is present under \(\HAlt\) but not under \(\HNul\) \cite{Efr:10}. In contrast to the traditional two-groups model, our proposed model captures the inhomogeneity of the underlying signal over $\jointDom$ by letting both, the prior null proportion \(\pi_0\circ\signalNodeEpoch\) and the alternative \gls{pdf} vary as a function of $\signalNodeEpoch$. Such inhomogeneity is often encountered in practice, an illustrative example is presented in \cref{exam:moving_trans} and \cref{fig:inhomo_illus}. 

With the proposed model, the conditional joint \gls{pdf} of the $p$-value $\mathrm{P}$, the hypothesis \(\mathrm{H}\), and the sample point $\tupleRV$ given \(\signal\) is 
\begin{align}\label{eq:pdf_pHvt|gam}
f_{\mathrm{P}, \mathrm{H}, \tupleRV}\big(p,H,(v,t)\mid \signal\big) =\rho(v,t)\cdot 
&\big(\pi_0\circ\signalNodeEpoch\cdot\idc{\HTrue=\HNul}\cdot \fPNulNodeEpoch\\
&\quad+ \big(1-\pi_0\circ\signalNodeEpoch\big)\cdot\idc{\HTrue=\HAlt}\cdot \fPAltNodeEpoch \big),\nonumber
\end{align}
where $p\in[0,1]$, $H\in\{\HNul, \HAlt\}$, and $\tuple\in\jointDom$. From \cref{eq:pdf_pHvt|gam}, we derive the conditional \gls{pdf} of $\pRVNodeEpochRV$ and $\tupleRV$ as
\begin{align}\label{eq:pdf_pvt|gam}
f_{\mathrm{P}, \tupleRV}\big(p,\tuple\mid\signal\big)
&= \rho(v,t)\cdot \big(\pi_0\circ\signalNodeEpoch\cdot \fPNulNodeEpoch + \big(1-\pi_0\circ\signalNodeEpoch\big)\cdot \fPAltNodeEpoch \big),
\end{align}
and the conditional \gls{pdf} of $\pRV = \mathrm{P}\tuple$ given $\signal$ and $\tupleRV = \tuple$ as
\begin{align}\label{eq:pdf_p|gamvt}
\fPNodeEpoch \equiv \fP[p\mid\signalNodeEpoch] &= \pi_0\circ\signalNodeEpoch\cdot \fPNulNodeEpoch + \big(1-\pi_0\circ\signalNodeEpoch\big)\cdot \fPAltNodeEpoch.
\end{align}

In other words, we parameterize the distribution of $p$-values on $\jointDom$ by a \emph{random generalized graph signal} $\signal$ (cf.\ \cref{sect:prelim:GGSP}). This model combines the information from the joint domain $\jointDom$ via $\signal$ and the two-groups model.
We illustrate our model in the following example.

\begin{Example}\label{exam:moving_trans}
\begin{figure}[!htb]
	\centering
	\begin{subfigure}[b]{\textwidth}
		\centering
		\includegraphics[width=0.44\linewidth, trim=3.5cm 1.5cm 3.2cm 4.2cm, clip]{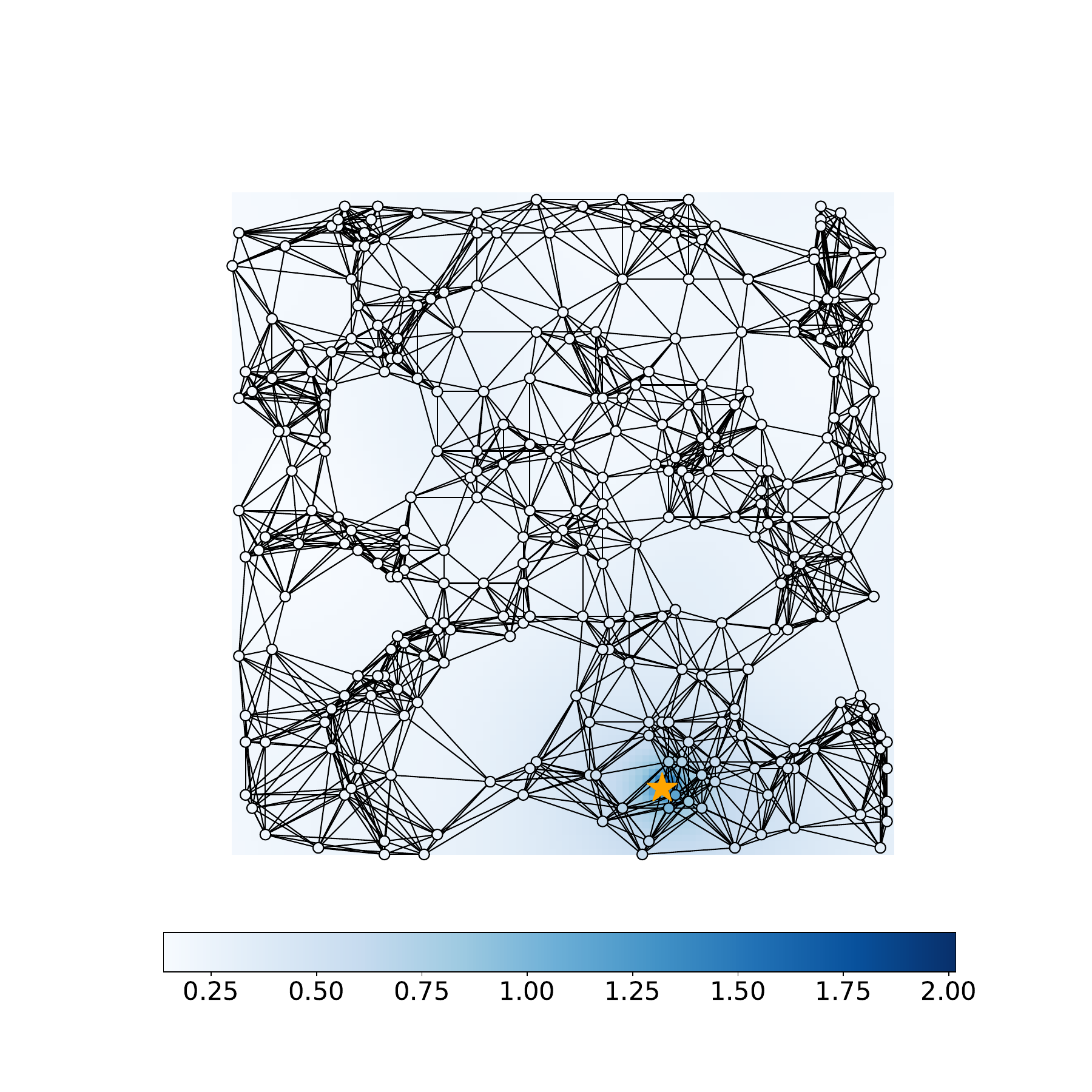}
        \includegraphics[width=0.44\linewidth, trim=3.5cm 1.5cm 3.2cm 4.2cm, clip]{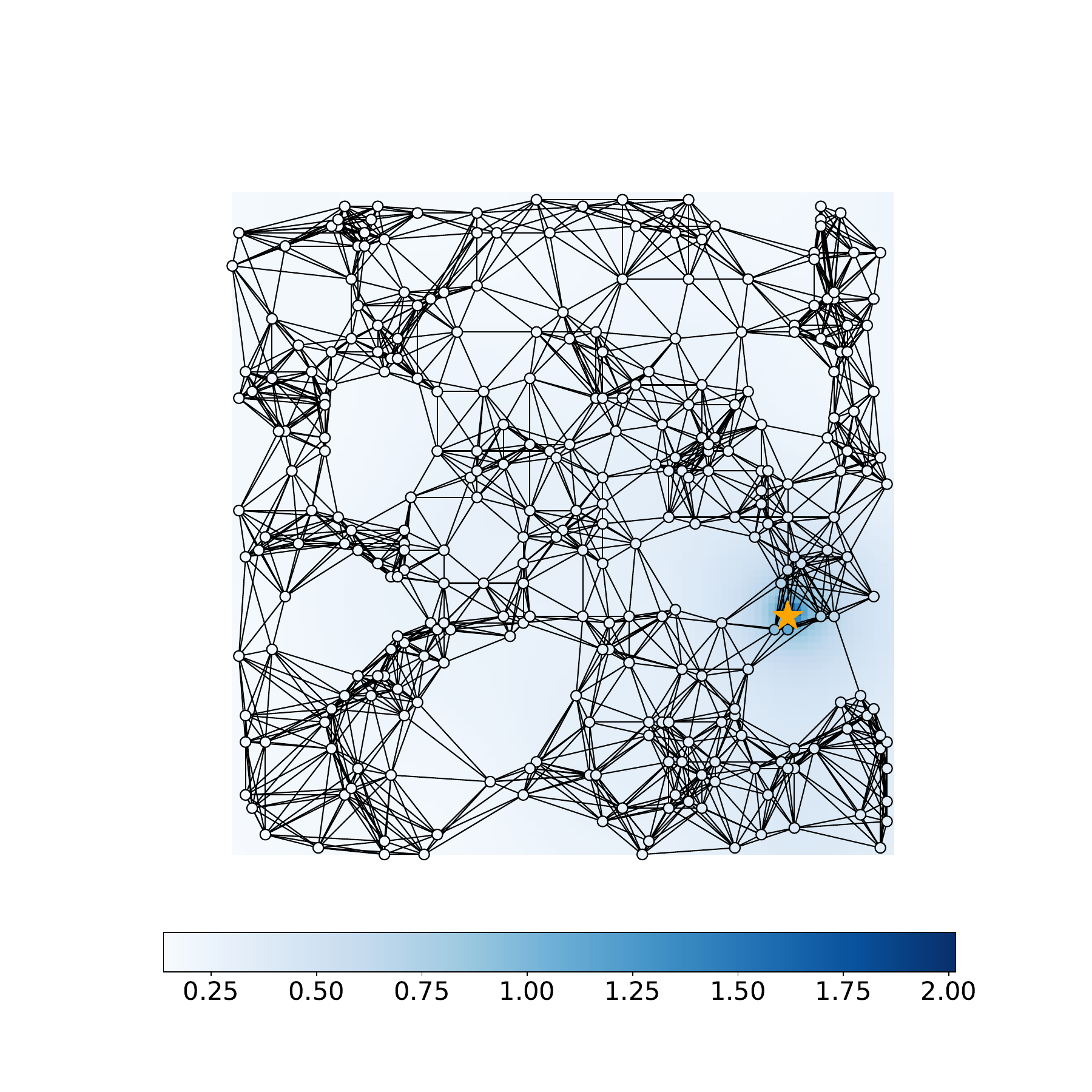}
		\caption{The generalized graph signal $\signal_{\tuple} = (\rVar{x}_{\text{sf}}(v,t,i))_{i=1,2}$. Here, we show the snapshot of $\signal_{\tuple}$ at a fixed time instance $t$. For clarity of visualization, the first image shows $\rVar{x}_{\text{sf}}(v,t,1)^{1/2}$, and the second shows $\rVar{x}_{\text{sf}}(v,t,2)^{1/2}$. The transmitters' positions are highlighted as stars in orange.}
		\label{fig:gamma_illustrate}
	\end{subfigure}
	\begin{subfigure}[b]{\textwidth}
		\centering
		\includegraphics[width=0.44\linewidth, trim=0.2cm 0.3cm 0.2cm 0.2cm, clip]{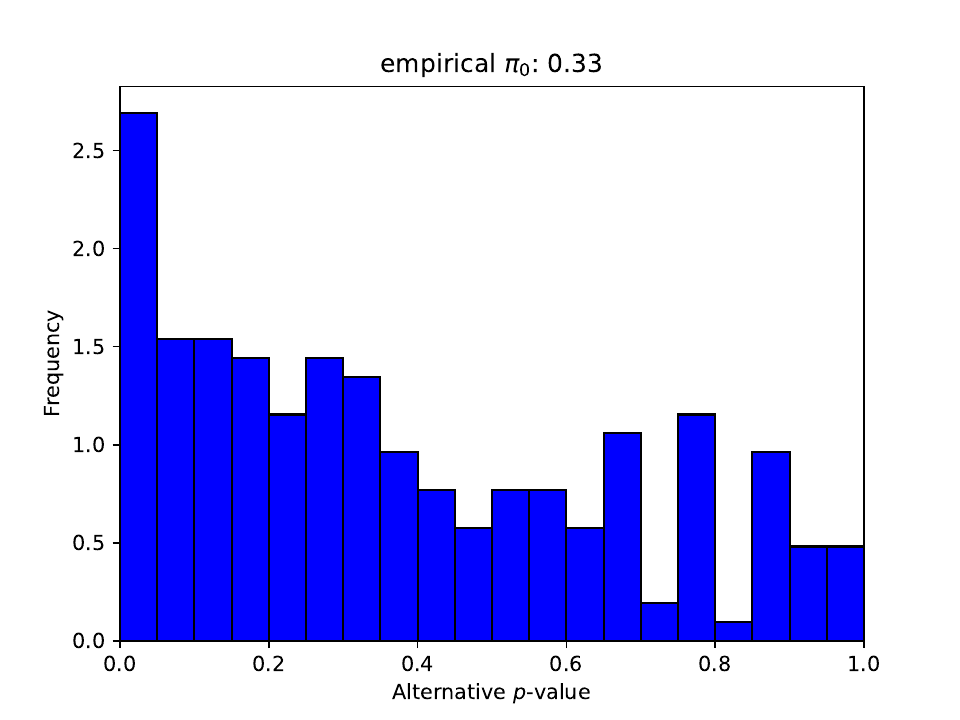}
    \includegraphics[width=0.44\linewidth, trim=0.2cm 0.3cm 0.2cm 0.2cm, clip]{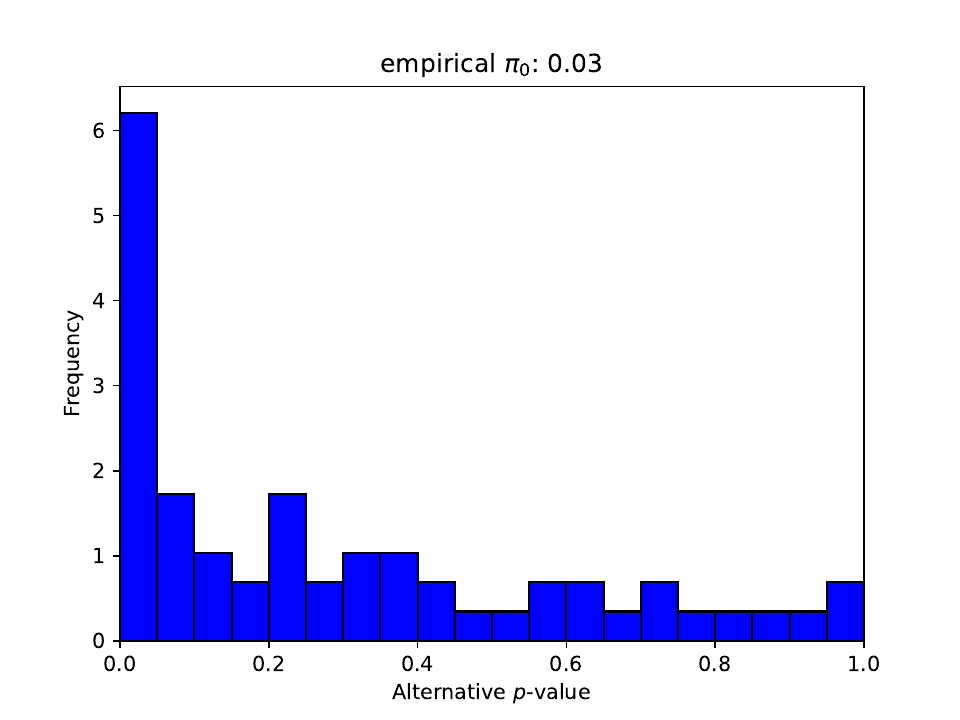}
		\caption{The left figure shows the empirical $\pi_0$ and the empirical histogram corresponding to $\fPAlt$ when $\norm{\signal_{\tuple}}_2\in[0.2, 0.35)$. The right figure shows these quantities when $\norm{\signal_{\tuple}}_2\in[0.35, 0.5)$. }
		\label{fig:p_dist_illustrate}
	\end{subfigure}
	\caption{Illustration of $\signal_{\tuple}$ in \cref{exam:moving_trans} and how the proportion of null hypotheses and the empirical distribution of $p$-values from alternatives vary with $\signal_{\tuple}$.}
	\label{fig:inhomo_illus}
\end{figure}

Consider a radio signal emitted by multiple transmitters in a 2-dimensional (2D) region. It is monitored by a sensor network \(\graph\). We first introduce the signal propagation model from a single transmitter. Denote the location of the transmitter as $\bc = (c_x, c_y)$. Let the signal magnitude (i.e., the absolute value of the signal), say $x_0$, be transmitted. We denote the distance between a node $\node$ (with coordinates $(v_x,v_y)$) and the transmitter as $d(\node,\bc)$. The signal value received by each node $v$ from this single transmitter, denoted by $\rVar{x}(v)$, is subject to path loss, shadow fading and fast fading \cite{CaiGia:J03,JakCox:94}: 
\begin{align}
    \rVar{x}_{\text{sf}}(v) &= Cx_0 \cdot \frac{\lambda}{4\pi d(v,\bc)} \cdot \exp(\rVar{s}(v_x,v_y)), \label{eq:paloss-shafa}\\
    \rVar{x}_{\text{ff}}(v) &\sim 
    \begin{cases}
        \mathrm{Rice}(\mathrm{ratio}, \rVar{x}_{\text{sf}}(v)^2) & \text{if LOS path dominates},\\
        \mathrm{Rayleigh}(\rVar{x}_{\text{sf}}(v)^2) & \text{otherwise}.
    \end{cases}\label{eq:fafa}
\end{align}

In \cref{eq:paloss-shafa}, $\rVar{s}(\cdot,\cdot)$ is a two-dimensional Gaussian process, $\lambda$ is the wavelength, and $C$ is a constant. The quantity $\rVar{x}_{\text{sf}}(v)$ denotes the signal magnitude after path loss and shadow fading.   In \cref{eq:fafa}, we present two typical types of fast fading. In cases where there is a dominant \gls{LOS} path, the received signal magnitude follows a Rician distribution where $\mathrm{ratio}$ represents the ratio between signal power from the \gls{LOS} and the remaining multipath. Another parameter equals to the signal power $\rVar{x}_{\text{sf}}(v)^2$. On the other hand, when there is no single dominant signal path, the received signal magnitude follows a Rayleigh distribution where the parameter is the signal power.

When there are multiple moving transmitters, we denote the signal that node $v$ receives from transmitter $i$ at time $t$ as $\rVar{x}_{\text{ff}}(v,t,i)$. We assume that node $v$ receives the strongest signal. i.e., 
\begin{align}\label{eq:max_sig}
    \abs{\rVar{x}(v, t)} = \max_{i} \abs{\rVar{x}_{\text{ff}}(v,t,i)}.
\end{align}
Finally, the signal energy measured by the node $v$ at instance $t$ is
\begin{align*}
\rVar{y}(v, t) := \rVar{x}(v, t) + \rVar{e},
\end{align*}
where $\rVar{e}$ is \gls{AWGN}. 


Suppose that we wish to test whether the received signal energy $\rVar{x}(v,t)$ is above the noise floor $\tau_0$, i.e., 
\begin{align*}
\HNul &: \abs{\rVar{x}(v,t)}\leq \tau_0,\\ 
\HAlt &: \abs{\rVar{x}(v,t)} > \tau_0,
\end{align*}
then the summary statistics at $(v,t)$ is $\p=\bbP(\rVar{Y}^2\geq \rVar{y}(v,t)^2\given \HNul)$. Here, $\rVar{Y}^2$ is a chi-squared random variable whose degree of freedom is $1$.


In this propagation model, the vector $(\rVar{x}_{\text{sf}}(v,t,i))_{i\geq 1}$ corresponds to the stochastic process $\signal_{\tuple}$. Given this $\signal_{\tuple}$, the distribution of $\rVar{x}(v,t)$ is determined by \cref{eq:fafa,eq:max_sig}. Besides, According to the nature of fast fading, $\rVar{x}(v,t)$ for different $(v,t)$ are independent given $\signal$. Hence, $\HTrue$ are independent with inhomogeneous probabilities on $\jointDom$ determined by $\signal$. When $\HTrue=\HNul$, the amplitude of $\rVar{x}(v,t)$ is small if $\tau_0$ is chosen to be small, and it can be approximated by $0$. In this case, the distribution of $\p$ under $\HNul$ is the uniform distribution on the interval $[0,1]$, $\dist{Unif}[0,1]$. When $\HTrue=\HAlt$, the distribution of $\p$ under $\HAlt$ only relies on the value of $\signal_{\tuple}$. This coincides with the hierarchical Bayesian model. We illustrate the dependence of the null proportion and the alternative $p$-value distribution on $\signal_{\tuple}$ in \cref{fig:inhomo_illus}. 
\end{Example}


In this paper, we assume that $\pi_0\circ\signalNodeEpoch$ and $\fPAltNodeEpoch$ are identifiable from $\fPNodeEpoch$ using \cref{eq:pdf_p|gamvt}. This can be guaranteed by the following assumptions:
\begin{Assumption}\label{asp:identify}\
\begin{enumerate}[(i)]
\item \label[condition]{cond:monot_f_0} $\fPNulNodeEpoch$ is non-decreasing \gls{wrt} $p\in[0,1]$ for all $\tuple\in 
\jointDom$.
\item \label[condition]{cond:monot_f_1}$\fPAlt[p;\zeta]$ is non-increasing \gls{wrt} $p\in(0,1]$ for all $\zeta\in\calZ$.
\item \label[condition]{cond:f_1=0}$\min\limits_{p\in[0,1]}\fPAlt[p;\zeta]=0$ for all $\zeta\in\calZ$, i.e., with \cref{cond:monot_f_1}, $\fPAlt[1;\zeta]=0$ for all $\zeta\in\calZ$.
\item \label[condition]{cond:conts_pdf}$\fPNulNodeEpoch$ is continuous on $[0,1]\times\jointDom$, and $\fPAlt[p;\zeta]$ is continuous on $(0,1]\times\calZ$.
\end{enumerate}
\end{Assumption}

In \cref{asp:identify}, \cref{cond:monot_f_0,cond:monot_f_1} indicate that the $p$-value is more likely to be small under the alternative and more likely to be large under the null hypothesis. 
These conditions are commonly assumed in the \gls{MHT} literature (cf.\ \cite[Theorem 2]{LeiFit:J18} and \cite[Section 2.1]{CaoChenZhang:J22}).
They ensure that $\pi_0\circ\signalNodeEpoch$ and $\fPAltNodeEpoch$ are identifiable from $\fPNodeEpoch$, since according to \cref{eq:pdf_p|gamvt} and \cref{cond:monot_f_0,cond:monot_f_1,cond:f_1=0}, for each $\tuple\in\jointDom$, we have  
\begin{align}
\pi_0\circ\signalNodeEpoch &= \frac{\fPNodeEpoch[1]}{\fPNulNodeEpoch[1]},\label{eq:pi0_inv}\\
\fPAltNodeEpoch &= \ofrac{1-\pi_0\circ\signalNodeEpoch} \Big(\fPNodeEpoch - \pi_0\circ\signalNodeEpoch \fPNulNodeEpoch) \Big).\label{eq:f1_inv}
\end{align}

The \gls{lfdr} is defined as 
\begin{align}\label{lfdr}
\lfdr(p; \signal_{\tuple}) = \frac{\pi_0\circ\signalNodeEpoch \fPNulNodeEpoch}{\fPNodeEpoch}.
\end{align}
As in \cref{eq:lfdrnonlocal}, this is the conditional probability $\bbP(\HTrue = \HNul\given \signalNodeEpoch,\p = p)$. The difference here is that we are now using both localized \glspl{pdf} and null proportion in computing this conditional probability. In \cref{subsect:inference}, we estimate $\fPNodeEpoch$, and thus $\pi_0\circ\signalNodeEpoch$ and $\fPAltNodeEpoch$.

\section{Asymptotic FDR Control Approach}\label{sect:MHT-GGSP}

In this section, we first explain the detection strategy when the true values of \gls{lfdr} are known. Then, we propose a method to estimate the \gls{lfdr} and present a theoretical guarantee of \gls{FDR} control by using the estimated \gls{lfdr} for detection.

\subsection{Oracle solution}

In this subsection, we show that, if $\pi_0\circ\signalNodeEpoch$ and $\fPAltNodeEpoch$ are known for $\tuple\in\samReg$, then the optimal detection strategy is thresholding the \gls{lfdr}. To see this, we introduce two complementary definitions, the marginal \gls{FDR} and the marginal power, defined as 
\begin{align}
\mFDR(h;\signal,\samReg) &= \frac{\bbE[\abs{\altRegEst\bigcap\nulReg} \given \signal(\sfS)]}{\bbE[\abs{\altRegEst}\given\signal(\sfS)]},\label{eq:def_mfdr}\\
\mpow(h;\signal,\samReg) &= \frac{\bbE[\abs{\altRegEst\bigcap \altReg} \given \signal(\sfS)]}{\bbE[\abs{\altReg}\given\signal(\sfS)]}.\label{eq:def_mpow}
\end{align}

When the hypothesis tests are conducted separately, the null hypotheses are rejected if the $p$-value is lower than the pre-determined significance level. In this paper, for each individual test $\HTrue$, we use the following thresholding rule:
\begin{align}\label{eq:thres_rule}
h(\pVec)\tuple = 
\begin{cases}
\HAlt & \text{if}\ \p \leq s_{\tuple}, \\
\HNul & \text{otherwise}.
\end{cases}
\end{align}
The difference is that, in \cref{eq:thres_rule}, the threshold of $p$-value is no longer the significance level but chosen by certain criteria that we will soon introduce. Under this rule, designing the detection strategy $h$ amounts to designing $\boldsymbol{s}:= (s_{\tuple})_{\tuple\in\samReg}\in[0,1]^\numTests$. Since the rejection rule $h$ is fully determined by $\boldsymbol{s}$, we alternatively write $\mFDR(h; \signal,\samReg)$ and $\mpow(h; \signal,\samReg)$ as $\mFDR(\boldsymbol{s}; \signal,\samReg)$ and $\mpow(\boldsymbol{s}; \signal,\samReg)$ under \cref{eq:thres_rule}. We thus consider the following problem:
\begin{align}\label{eq:opt_mpow}
\begin{aligned}
\max_{\boldsymbol{s}\in[0,1]^\numTests} \mpow(\rVec{s}; \signal,\samReg) \\
\ST \mFDR(\boldsymbol{s}; \signal,\samReg) \leq \alpha.
\end{aligned}
\end{align}
In \cref{thm:opt_thres}, we show that the optimal solution to problem \cref{eq:opt_mpow} is a level surface of \gls{lfdr}. This theorem slightly generalizes \cite[Theorem 2]{LeiFit:J18}. In \cref{subsect:inference}, we show that by solving an approximate version of \cref{eq:opt_mpow}, we can obtain an $h$ that yields asymptotic \gls{FDR} control by a pre-determined \gls{FDR} level.

\begin{Theorem}\label{thm:opt_thres}
Suppose \cref{asp:identify} holds and there exists $p\in(0,1)$ and $\tuple\in\samReg$ such that $\lfdr(p;\signal_{\tuple})< \alpha$.
Then the optimal solution $\boldsymbol{s}^*$ to problem \cref{eq:opt_mpow} satisfies
\begin{align*}
\lfdr(s^*_{\tuple};\signal_{\tuple}) = \eta,\ m=1,\dots,M,
\end{align*}
where $\eta$ is a constant independent of $m$.
\end{Theorem}

Different from \cite[Theorem 2]{LeiFit:J18} and \cite[Proposition 2.1]{CaoChenZhang:J22}, in \cref{thm:opt_thres}, we do not require that $\fPAltNodeEpoch$ is continuous on the \emph{closed interval} $[0,1]$. This means that we allow for unbounded $f_1$ such as the Beta distribution. \cref{thm:opt_thres} implies that the optimal threshold for $\pVec$ corresponds to a level set of \gls{lfdr}. Therefore, the rejection rule becomes  
\begin{align}\label{eq:orac_thres}
h(\pVec)\tuple=
\begin{cases}
\HAlt & \lfdr(\p; \signal_{\tuple}) \leq \eta, \\
\HNul & \lfdr(\p; \signal_{\tuple}) > \eta.
\end{cases}
\end{align}
Note that in practice, we usually do not have access to the ground truth $\pi_0\circ\signalNodeEpoch$ and $\fPAltNodeEpoch$. Therefore, we call \cref{eq:orac_thres} an \emph{oracle} rejection rule. In the next section, we replace \cref{eq:orac_thres} by its estimate determined from samples. Here, we first explain the choice of $\eta$ under this oracle rule, and the sample-based version will then easily follow. To clarify the choice of $\eta$, we define the following quantities as in \cite{CaoChenZhang:J22}:
\begin{align}
\rVar{d}_{1,\numTests}(\eta) &:= \ofrac{\numTests}\sum_{\tuple\in\samReg}\lfdr(\p;\signalNodeEpoch)\idc{\lfdr(\p;\signalNodeEpoch)\leq \eta},\label{eq:def:D_1M}\\
\rVar{d}'_{1,\numTests}(\eta) &:=\ofrac{\numTests}\sum_{\tuple\in\samReg} \idc{\HTrue=\HNul}\idc{\lfdr(\p;\signalNodeEpoch)\leq \eta}\label{eq:def:V_M}\\
\rVar{d}_{0,\numTests}(\eta) &:= \ofrac{\numTests}\sum_{\tuple
\in\samReg}\idc{\lfdr(\p;\signalNodeEpoch)\leq \eta}.\label{eq:def:D_0M}
\end{align}

By conditioning on $\p$ and using the definition $\bbP(\HTrue=\HNul\given \p,\signalNodeEpoch) = \lfdr(\p;\signal_{\tuple})$, we see that $\rVar{d}_{1,\numTests}(\eta)$ estimates the proportion of false rejections among all tests:
\begin{align*}
&\bbE[\idc{\HTrue=\HNul}\idc{\lfdr(\p;\signalNodeEpoch)\leq \eta}\given\signal,\samReg]\\
&= \bbE[\lfdr(\p;\signalNodeEpoch)\idc{\lfdr(\p;\signalNodeEpoch)\leq \eta}\given\signal,\samReg].
\end{align*}
Therefore, by taking expectation over $\samReg$, we see that $\bbE[\rVar{d}_{1,\numTests}(\eta)\given\signal]=\bbE[\abs{\altRegEst\bigcap\nulReg}\given\signal]=\bbE[\rVar{d}'_{1,\numTests}(\eta)\given\signal]$. Besides, it can be shown that $\bbE[\rVar{d}_{0,\numTests}(\eta)\given\signal]=\bbE[\abs{\altRegEst}\given\signal]$. Therefore, the quantity
\begin{align*}
\rVar{r}_\numTests(\eta):=\frac{\rVar{d}_{1,\numTests}(\eta)}{\rVar{d}_{0,\numTests}(\eta)}
\end{align*}
approximates $\mFDR(h;\signal,\samReg)$ in \cref{eq:def_mfdr}. On the other hand, note that $\mpow(h;\signal,\samReg)$ in \cref{eq:def_mpow} increases with $\eta$. Therefore, we choose the optimal $\eta$ under the oracle rule \cref{eq:orac_thres} in the following way:
\begin{align}\label{eq:orac_sol}
\rVar{\eta}_\numTests:=\sup\set{\eta\given \rVar{r}_\numTests(\eta)\leq\alpha}.
\end{align}

\subsection{Joint density estimation and testing procedure}\label{subsect:inference}

In this subsection, we propose a method to estimate the unknown densities $\fPNodeEpoch$ (or equivalently $\fP[p\mid\signal\tuple]$) for $\tuple\in\samReg$, and then solve the \gls{MHT} problem with these estimates. Estimating the unknown densities is not an easy task in general since the number of unknown densities is the same as the number of $p$-values. In this paper, we take advantage of the \gls{GGSP} model to largely reduce the number of unknown parameters to a constant, so that the \gls{MLE} can be calculated and thus, the density estimation is consistent. This consistency then ensures asymptotic \gls{FDR} control. 

To ensure the consistency of the \gls{MLE}, we make the following assumptions:
\begin{Assumption}\label{asp:bandlimit}\
\begin{enumerate}[(i)]
\item The signal $\signal$ is bandlimited, i.e., for all $(v,t)\in\jointDom$,
\begin{align}\label{eq:bl_para}
\signal(v,t) = \sum_{k_1=1}^{K_1}\sum_{k_2=1}^{K_2} \upxi_{k_1,k_2} \cdot\phi_{k_1}(v)\psi_{k_2}(t),
\end{align}
where $K_1, K_2$ are known positive integers, $\set{\phi_k(v)\given k=1,\dots,N}$ is the graph Fourier basis and $\set{\psi_k(t)\given k\in\bbN}$ is a set of orthonormal basis of $L^2(\calT)$. The coefficient matrix $\bXi:=(\upxi_{k_1,k_2})\in\Real^{K_1\times K_2}$ is a random matrix and takes values in a convex and compact set $\calK\subset\Real^{K_1\times K_2}$. Under this assumption, we may write $\signal(v,t)$ as $\signal(v,t;\bXi)$ to highlight the relationship \cref{eq:bl_para}.

\item \label[condition]{cond:cont_basis} The function $\phi_{k_1}(v)\psi_{k_2}(t)$ is continuous in $(v,t)\in\jointDom$ for all $1\leq k_1\leq K_1,1\leq k_2\leq K_2$.

\item \label[condition]{cond:identi_pvt} For any distinct $\mathit{\Xi}\neq\mathit{\Xi'}$, $\fP[p\mid\signal(v,t;\mathit{\Xi})]\neq \fP[p\mid\signal(v,t;\mathit{\Xi'})]$ on a set in $(0,1]\times\jointDom$ with positive measure. We denote the distribution and expectation over $\pRVNodeEpochRV,\tupleRV$ given $\bXi=\mathit{\Xi}$ by $\bbP_{\mathit{\Xi}}$ and $\bbE_{\mathit{\Xi}}$.

\item For all $\mathit{\Xi}$,
\begin{align}
\bbE_{\mathit{\Xi}}[\ln\norm*{\dfrac{\fP[\pRVNodeEpochRV\mid\signal(\nodeRV,\epochRV;\mathit{\Xi'})]}{\fP[\pRVNodeEpochRV\mid\signal(\nodeRV,\epochRV;\mathit{\Xi})]}}_\infty]<\infty,
\end{align}
where the sup norm is taken \gls{wrt} $\mathit{\Xi'}$.
\end{enumerate}
\end{Assumption}

Under \cref{asp:bandlimit}, the number of variables to be estimated is $K_1K_2$, which is independent of $\numTests$. The \gls{MLE} can thus be calculated by maximizing the log-likelihood function:
\begin{align*}
\max_{\mathit{\Xi}\in\calK}\sum_{\tuple\in\samReg}l(\mathit{\Xi};\p,\tuple)
\end{align*}
where 
\begin{align*}
l(\mathit{\Xi};\p,\tuple) = \ln \fP[\p\mid\signal(v,t;\mathit{\Xi})]  + \ln\rho\tuple.
\end{align*}
Note that $\rho\tuple$ does not depend on $\mathit{\Xi}$, so the \gls{MLE} $\widehat{\bXi}$ can be obtained by 
\begin{align}\label{eq:MLE_prob}
\argmax_{\mathit{\Xi}\in\calK} &\sum_{\tuple\in\samReg} \ln \fP[\p\mid\signal(v,t;\mathit{\Xi})].
\end{align}
Recall that $\calK$ is a convex and compact set. In practice, since the ground truth is not known in advance, we recommend defining $\calK$ as a sufficiently large set, e.g.,  $\calK = \set{\mathit{\Xi}\given \mathit{\Xi}_{k_1,k_2}\in [-c_{\max}, c_{\max}]}$ where $c_{\max}$ is the largest finite real number recognized by the computer. By the consistency of \glspl{MLE}, $\widehat{\bXi}$ converges to $\bXi$ in probability as $\numTests\to\infty$, hence 
\begin{align*}
\widehat{\signal}(v,t) = \sum_{k_1=1}^{K_1}\sum_{k_2=1}^{K_2} \widehat{\upxi}_{k_1,k_2} \cdot\phi_{k_1}(v)\psi_{k_2}(t)
\end{align*}
converges to $\signal(v,t)$ in probability. We formally state this in the following theorem:

\begin{Theorem}\label{thm:unif_conv_para}
Under \cref{asp:bandlimit}, we have $\widehat{\bXi}\convp\bXi$ as the number of samples $\numTests\to\infty$ and
\begin{align}\label{eq:det_consis_mle}
\sup_{(v,t)\in\jointDom}\abs{\widehat{\signal}(v,t) - \signal(v,t)}\convp0,
\end{align}
under the probability measure conditioned on $\bXi$.
\end{Theorem}

\Cref{thm:unif_conv_para} indicates that the \gls{MLE} of the parameter $\signal(v,t)$ uniformly converges on $\jointDom$. As we will see in \cref{thm:limit_FDR}, this property ensures the asymptotic control of \gls{FDR}, hence justifies the usage of \gls{MLE} in the inference of \glspl{pdf}.

In practice, we need to choose the hyperparameters $K_1$ and $K_2$ to balance the goodness of fit and model complexity. Let $l_{K_1,K_2}^*$ be the optimal value of \cref{eq:MLE_prob}. We propose to use the \gls{BIC} for choosing these parameters:
\begin{align*}
\mathrm{BIC} = K_1K_2\ln \numTests - 2l_{K_1,K_2}^*.
\end{align*}
Using \cref{thm:unif_conv_para}, we can estimate $\fP[p\mid\signal(v,t;\bXi)]$ by $\fP[p\mid\signal(v,t;\widehat{\bXi})]$. Using the estimation result, we infer $\pi_0\circ\signal(v,t)$ and $\fPAlt[p;\signal\tuple]$ according to \cref{eq:pi0_inv,eq:f1_inv}:
\begin{align}
\pi_0\circ\signal(v,t;\widehat{\bXi}) &= \frac{\fP[1\mid\signal(v,t;\widehat{\bXi})]}{\fPNulNodeEpoch[1]},\label{eq:pi0_est}\\
\fPAlt[p;\signal(v,t;\widehat{\bXi})] &= \ofrac{1-\pi_0\circ\signal(v,t;\widehat{\bXi})}(\fP[p\mid\signal(v,t;\widehat{\bXi})] - \pi_0\circ\signal(v,t;\widehat{\bXi})\fPNulNodeEpoch)).\label{eq:f1_est}    
\end{align}

In previous works addressing the \gls{MHT} problem on graphs, the optimization problems often involve a number of parameters equal to the number of $p$-values (cf.\ \cite[(5)]{LiBar:J18}, \cite[(6)]{TanKoyPolSco:J18}), making it challenging to solve these high-dimensional optimization problems. In contrast, our method only requires $K_1K_2$ parameters, which do not increase with $\numTests$ and are typically much smaller than $\numTests$.

Once we have estimated $\signal(v,t)$,  we can then estimate \gls{lfdr}, $\rVar{d}_{1,\numTests}(\eta)$, $\rVar{d}_{0,\numTests}(\eta)$ and $\rVar{r}_\numTests(\eta)$:
\begin{align*}
\lfdr(p;\signal(v,t;\widehat{\bXi})) &:= \frac{\pi_0\circ\signal(v,t;\widehat{\bXi})\fPNulNodeEpoch}{\fP[p;\signal(v,t;\widehat{\bXi})]},\\
\widehat{\rVar{d}}_{1,\numTests}(\eta) &:= \ofrac{\numTests}\sum_{\tuple\in\samReg}\lfdr(\p;\signal(v,t;\widehat{\bXi}))\idc{\lfdr(\p;\signal(v,t;\widehat{\bXi}))\leq \eta},\\
\widehat{\rVar{d}'}_{1,\numTests}(\eta)&:=\ofrac{\numTests}\sum_{\tuple\in\samReg}\idc{\HTrue=\HNul}\idc{\lfdr(\p;\signal(v,t;\widehat{\bXi}))\leq \eta},\\
\widehat{\rVar{d}}_{0,\numTests}(\eta) &:= \ofrac{\numTests}\sum_{\tuple\in\samReg}\idc{\lfdr(\p;\signal(v,t;\widehat{\bXi}))\leq \eta},\\
\widehat{\rVar{r}}_\numTests(\eta)&:=\frac{\widehat{\rVar{d}}_{1,\numTests}(\eta)}{\widehat{\rVar{d}}_{0,\numTests}(\eta)}.
\end{align*}
Therefore, by \cref{eq:orac_sol}, we design the rejection threshold $s_{\tuple}$ such that
\begin{align*}
\lfdr(s_{\tuple};\signal(v,t;\widehat{\bXi})) = \widehat{\upeta}_\numTests,
\end{align*}
where
\begin{align*}
\widehat{\upeta}_\numTests:=\sup\set{\eta\given\widehat{\rVar{r}}_\numTests(\eta)\leq\alpha}.
\end{align*}

In the remainder of this paper, we write $h$ to denote the thresholding strategy \cref{eq:orac_thres} with $\eta=\widehat{\upeta}_\numTests$. We call this method MHT-GGSP. To achieve asymptotic \gls{FDR} control, we make the following regularity assumptions:
\begin{Assumption}\label{asp:regularity}
Assume that the following conditions hold:
\begin{enumerate}[(i)]
\item \label[condition]{cond:monotone_pdf}$\fPAlt[p;\zeta]$ is strictly decreasing on $p\in(0,1]$ for all $\zeta$. 
\item \label[condition]{cond:pi0_pos}$1>\pi_0\circ\signal(v,t;\mathit{\Xi})>0$ for all $(v,t)\in\jointDom$ and $\mathit{\Xi}\in\calK$.
\item \label[condition]{cond:f0_pos}$\fPNulNodeEpoch>0$ whenever $p>0$.
\item \label[condition]{cond:conts_limit_lfdr} Let $\chi(v,t;\bXi):=\lim\limits_{p\to0^+}\lfdr(p;\signal(v,t;\bXi))$. Then $\chi(v,t;\bXi)$ is always continuous on $\jointDom$. 
\item \label[condition]{cond:small_lfdr} There always exists $(v,t)\in\jointDom$ and $p>0$ such that $\lfdr(p;\signal(v,t;\bXi))<\alpha$.
\item \label[condition]{cond:conts_pd} 
Rewriting \cref{eq:pdf_p|gamvt}, with $\zeta=\signal(v,t)$, as
\begin{align*}
f'_{\rVar{P}}(p\mid\zeta,(v,t)) = \pi_0(\zeta)\fPNulNodeEpoch + (1-\pi_0(\zeta))\fPAlt[p;\zeta],
\end{align*}
we assume $\dfrac{\partial f'_{\rVar{P}}}{\partial\zeta}(p\mid\zeta,(v,t))$ is continuous on $(0,1]\times\calZ\times\jointDom$.
\end{enumerate}
\end{Assumption}

In \cref{asp:regularity}, \cref{cond:monotone_pdf} is a slightly stronger condition than \cref{cond:monot_f_1} in \cref{asp:identify}. \Cref{cond:pi0_pos} ensures that the Bayesian model in \cref{sect:stat_model} is non-trivial, i.e., condition on any $\signal$, $\HTrue$ is always random. \Cref{cond:f0_pos} states that it is always possible to observe small $p$-values under null hypotheses, which is a prevalent phenomenon in hypothesis testing. \Cref{cond:conts_limit_lfdr,cond:small_lfdr} assume good identifiability of alternative when $p$ is small enough. Combining these two conditions, we know that there exists a non-empty open set in $\jointDom$ such that for any $(v,t)$ in this open set, there exists $p$ such that $\lfdr(p;\signal(v,t;\bXi))<\alpha$. Since $\rho$ is a strictly positive measure, This implies that the probability that the assumption in \cref{thm:opt_thres} holds tends to $1$ as $\numTests$ tends to infinity. These conditions hold true, for example, when $\lim\limits_{p\to 0} \fPAltNodeEpoch = \infty$ and $\fPNulNodeEpoch$ is bounded on $[0,1]$ for all $(v,t)$, we have $\chi(v,t;\bXi)=0$ for all $(v,t)$.  \Cref{cond:conts_pd} is a smoothness assumption on $f'_{\rVar{P}}$. 

To state the result on asymptotic \gls{FDR} control, we define the following quantities:
\begin{align*}
\rVar{d}_0(\eta) &:= \bbP(\lfdr(\pRVNodeEpochRV;\signal(\nodeRV,\epochRV;\bXi))\leq\eta|\bXi),\\
\rVar{d}_1(\eta) &:= \bbE[\idc{\HTrueRV = \HNul}\idc{ \lfdr(\pRVNodeEpochRV;\signal(\nodeRV,\epochRV;\bXi))\leq\eta}\given\bXi],\\
\rVar{d}_2(\eta)&:=\rVar{d}_0(\eta) - \rVar{d}_1(\eta), \\
\rVar{r}(\eta)&:=\dfrac{\rVar{d}_1(\eta)}{\rVar{d}_0(\eta)},\\
\upkappa_0&:=\bbE[\idc{\HTrueRV = \HNul}\given\bXi].
\end{align*}
Under the aforementioned assumptions, we have the following results:
\begin{Theorem}\label{thm:limit_FDR}
Under \cref{asp:identify,asp:bandlimit,asp:regularity}, we have
\begin{align*}
\uplim_{\numTests\to\infty}\FDR_{\samReg}(h)\leq\alpha. 
\end{align*}
\end{Theorem}

\begin{Theorem}\label{thm:limit_pow}
Suppose \cref{asp:identify,asp:bandlimit,asp:regularity} hold. 
Let  $\upeta_0:=\sup\set{\eta\given \rVar{r}(\eta)\leq\alpha}$. We have
\begin{align*}
\pow_{\samReg}(h)\convp\frac{\rVar{d}_2(\upeta_0)}{1-\upkappa_0}\ \text{as}\ \numTests\to\infty.
\end{align*}
\end{Theorem}

We remark that our result is more general than existing results for \gls{FDR} control over graphs. In \cite[Proposition 3]{PouXia:J23}, the authors prove asymptotic \gls{FDR} control under the assumption that the $p$-value distribution under the alternative is homogeneous among the vertices. When this homogeneity assumption does not strictly hold, the \gls{FDR} upper bound shows inflation and depends on the deviation of the proportions of null hypotheses among different vertices \cite[Theorem 2]{PouXia:J23}. By utilizing the domain information of $\jointDom$, we allow the $p$-value distribution under the alternative to vary in the joint domain while still maintaining \gls{FDR} control. In \cite[Lemma 3]{LiBar:J18}, the asymptotic control of \gls{FDR} simultaneously depends on the norm of the incidence matrix of \(\graph\) and the sparsity level of $p$-values' weights, which restricts the user's choices of sparsity level, making \gls{FDR} control inaccessible in some cases. In contrast, the asymptotic control of \gls{FDR} always holds for our approach, irrespective of $K_1$ and $K_2$.

\section{Numerical Results}\label{sect:exp}

In this section, we compare the following MHT methods on a communication sensor network dataset. The code to reproduce the results is available online.\footnote{\url{https://github.com/xcjian/GGSP-detection}}
\begin{enumerate}
\item MHT-GGSP: We set $\fP[p\mid\signal_{\tuple}] = \beta\tuple p^{\beta\tuple-1}$, where $\beta\tuple = (1 + \exp(-\signal_{\tuple}))^{-1}$. It can be shown using \cref{eq:pi0_inv} that $\beta\tuple = \pi_0\circ \signalNodeEpoch$. We use the graph Laplacian as the \gls{GSO}, and $\psi_k$ the trigonometric basis of $L^2[-\pi, \pi]$ (cf. \cref{eq:tri_basis}). Here $\signal_{\tuple}$ is estimated by \cref{eq:MLE_prob} as $\widehat{\signal}_{\tuple}$, which yields the estimates $\widehat{\beta}\tuple := (1 + \exp(-\widehat{\signal}_{\tuple}))^{-1}$. 

\item MHT-GGSP\_\text{reg}: This is a modified version of MHT-GGSP designed to maintain the marginal correctness of the inferred model. In this method, we impose the constraint that $\ofrac{\numTests}\sum_{\tuple\in\samReg}\pi_0\circ\signalNodeEpoch= \widetilde{\pi}_0$, where $\widetilde{\pi}_0 = \ofrac{(1 - \tau_0) \numTests }\abs{\set{\tuple\given \p\geq \tau_0}}$ is Storey’s estimator of the null proportion \cite{Sto:J02}. Here, we set $\tau_0 = 0.5$. This is achieved by regularizing the estimated $\beta\tuple$ from MHT-GGSP such that their average equals $\widetilde{\pi}_0$, since $\beta\tuple$ coincides with the null proportion in this case. Hence, the estimated parameters in this method are $\tilde{\beta}\tuple:=\frac{\widetilde{\pi}_0}{\sum_{\tuple\in\samReg} \widehat{\beta}\tuple}\widehat{\beta}\tuple$.

\item MHT-GGSP\_\text{cens}: This is a modified version of MHT-GGSP\_\text{reg} designed to address model mismatch issues under low noise levels. Under such conditions, the $p$-values for alternative hypotheses are very close to zero, almost following a delta distribution at zero, while those for the null hypotheses follow a $\mathrm{Unif}[0,1]$ distribution. This abrupt change in $p$-value distribution at the alternative-null boundary, rather than a smooth variation over the graph, leads to an underestimation of $\beta\tuple$ for null $p$-values near the boundary of the alternative region. 

To mitigate this underestimation, we employ a censoring strategy inspired by \cite{Goelz2023a}. We set a small threshold $\eta_0$, which is $10^{-4}$ in this experiment. $p$-values below this threshold are referred to as ``censored $p$-values'' and the rest as ``non-censored $p$-values''. For non-censored $p$-values, we use MHT-GGSP to estimate $\beta\tuple$ without considering the censored $p$-values, aiming to reduce parameter underestimation caused by these censored values. For censored $p$-values, we estimate their null proportion as $\widehat{\pi_0\circ\signalNodeEpoch} = \eta_0\numTests / \abs{\set{\tuple\given \p\leq \eta_0}}$, independent of their associated $\tuple$. This provides an upper bound on the expected proportion of null $p$-values in the interval $[0, \eta_0]$ given the number of censored $p$-values. We assume $\fPAltNodeEpoch$ to be a narrow uniform distribution on $[0, \eta_0]$, also independent of $\tuple$, so it suffices to estimate the mass it has on $[0, \eta_0]$. 

To achieve this, we use the estimated $\fPAltNodeEpoch$ for non-censored $p$-values and the marginal relationship $\bbP(\p\leq \eta_0\given\samReg) = \ofrac{\numTests} (\sum_{\tuple\in\samReg} \FPNodeEpoch[\eta_0] \Ind\set{\p>\eta_0} + \abs{\set{\tuple\given \p\leq \eta_0}}( \pi_0\circ\signalNodeEpoch\eta_0 + (1 - \pi_0\circ\signalNodeEpoch)\FPAltNodeEpoch[\eta_0]))$. The left-hand side can be estimated by the sample proportion of censored $p$-values. On the right-hand side, the only unknown and not estimated quantity is $\FPAltNodeEpoch[\eta_0]$, which can be estimated using this equation.

\item \gls{BH} \cite{BenHoc:J95, Efr:10}: This method sets adaptive threshold on $p$-values. Suppose the $p$-values are ordered as $\set{p_{(i)}\given i=1,\dots,M}$. This method rejects the $p$-values less than or equal to
\begin{align*}
\max\set*{p_{(i)}\given p_{(i)} \leq \frac{i}{M}\alpha}.
\end{align*}
\item lfdr-sMoM \cite{GolZouKoi:J22}: This method assumes that all the $p$-values can be divided into several groups such that in each group they follow the same marginal distribution $f_\mathrm{P}$ which is a Beta mixture distribution \cite[(12)]{GolZouKoi:J22}. The group assignment and the marginal distribution are inferred from the data. We implement this method using the original code \cite{web:lfdr-sMoM}.
\item Proportion-matching \cite{PouXia:J23}: This method assumes that $f_1$ is homogeneous over the graph, while the null proportion is different on each vertex. By adjusting the \gls{FDR} control levels on different vertices, and applying the \gls{BH} method on each vertex, it is expected to match the performance of using the global \gls{BH} method on all $p$-values from all vertices.
\item FDR-smoothing \cite{TanKoyPolSco:J18}: This method uses $z$-values. It assumes that $f_1$ is homogeneous over the graph, while the null proportion differs on each vertex. This method estimates $f_1$ and $\pi_0$ by combining the negative log-likelihood function with the penalty term being the $l_1$-smoothness of $\pi_0$. The graph is constructed by the time-vertex approach, i.e., the product of the graph with the cyclic graph. We implement this method using the original code \cite{web:FDR-smoothing}.
\item SABHA \cite{LiBar:J18}: This method first reweighs the $p$-values and then applies the \gls{BH} method on the weighted $p$-values. The weights are understood as $\pi_0$, and the reciprocals are assumed to come from a feasible set. In this experiment, the feasible set consists of the graph signals that have $l_1$ norm less than a predetermined threshold. The weights are obtained by solving the optimization problem \cite[(5)]{LiBar:J18}. The graph is constructed by the time-vertex approach. We implement this method using the original code \cite{web:SABHA}.
\item AdaPT \cite{LeiFit:J18}: This iterative strategy masks most of the $p$-values in the beginning. At each iteration it reveals a certain amount of $p$-values according to a threshold, estimates the \gls{FDR} and updates the threshold. It stops when the \gls{FDR} estimate is lower than the nominal \gls{FDR} level, and rejects the $p$-values below the threshold. The threshold of the $p$-values is updated by an EM estimate of \gls{lfdr}. We implement this method using its R package \cite{web:AdaPT}. This method makes use of the true coordinates of the sensors, instead of the graph structure. 
\end{enumerate}


Following the model in \cref{exam:moving_trans}, we consider a 2D area where two wireless transmitters perform random walks on a $100\times 100$ grid (cf. \ \cref{fig:instance_8_comm,fig:instance_9_comm}). Receivers are randomly placed on $300$ points of the grid. We model the receiver sensor network as a $10$-NN graph according to their coordinates, and $\calT=[-\pi,\pi]$. Each received signal is affected by path loss, shadow fading and fast fading. Besides, the observations at each receiver are corrupted by \gls{AWGN}. We suppose the sample set $\samReg$ is given by $\calV\times \set{-\pi+\frac{j}{T}\cdot2\pi\given j=0,\dots, T}$. In this experiment, we set $T=9$. We are interested in determining whether each node, at each time instance, has received a signal above the noise floor from at least one transmitter. We follow the setup of scC, Cnfg. 2 in \cite{GolZouKoi:J22}. The data is simulated using the source code in \cite{GolZouKoi:J22} with the following additional features: i) the transmitters perform random walks. ii) The Rician fading (cf.\ \cref{eq:fafa}) is considered in computing the received signal. iii) The noise energy of the \gls{AWGN} varies over the range of $\set{10^{-3}, 0.25, 0.5, 1.0, 1.25, 1.5, 1.75, 2.0}$. The proportion of null hypotheses is approximately $10\%$.
\begin{figure}[htbp]
	\centering
	\begin{subfigure}[b]{0.49\textwidth}
		\centering
		\includegraphics[width=\linewidth, trim=4.5cm 4.3cm 4cm 4cm, clip]{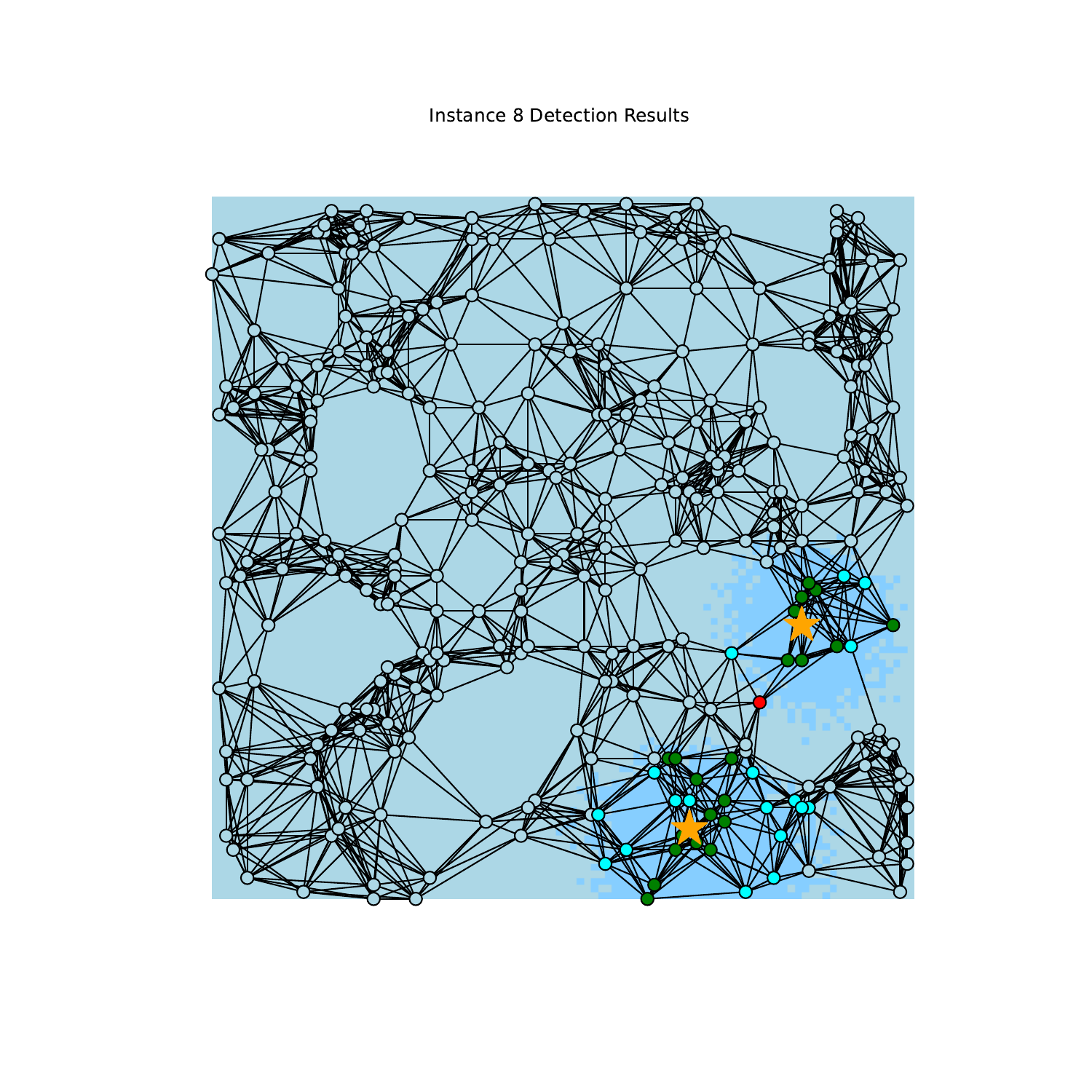}
		\caption{}
		\label{fig:instance_8_comm}
	\end{subfigure}
        \begin{subfigure}[b]{0.49\textwidth}
		\centering
		\includegraphics[width=\linewidth, trim=4.5cm 4.3cm 4cm 4cm, clip]{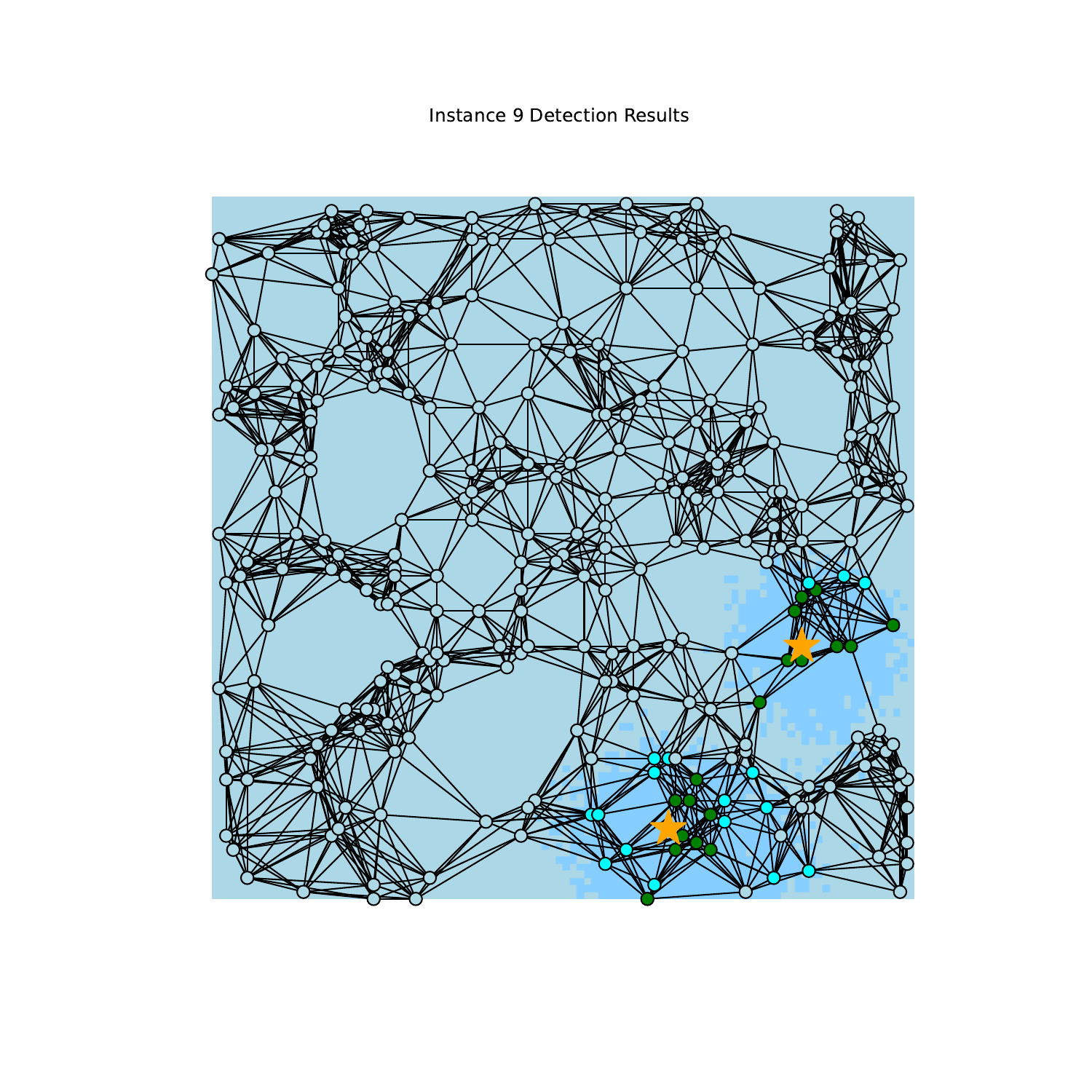}
		\caption{}
		\label{fig:instance_9_comm}
	\end{subfigure}
	\caption{Example of detection results by MHT-GGSP with nominal \gls{FDR} level 0.2, under noise level 1.5. \Cref{fig:instance_8_comm,fig:instance_9_comm} are detection results on the communication network for two consecutive instances. In the background, the light blue color denotes the null region, while a deeper color denotes the alternative region. On the graph, the light blue color represents correctly identified nulls, and the green color represents correctly rejected alternatives. Cyan represents undetected alternatives, and red represents incorrectly rejected nulls. We use orange stars to highlight the transmitters' locations. 
    }
	\label{fig:detect_illus}
\end{figure}

\begin{figure}[!htb]
	\centering
	\begin{subfigure}[b]{0.49\textwidth}
		\centering
		\includegraphics[width=0.95\linewidth, trim=0.55cm 0.3cm 1.5cm 1.4cm, clip]{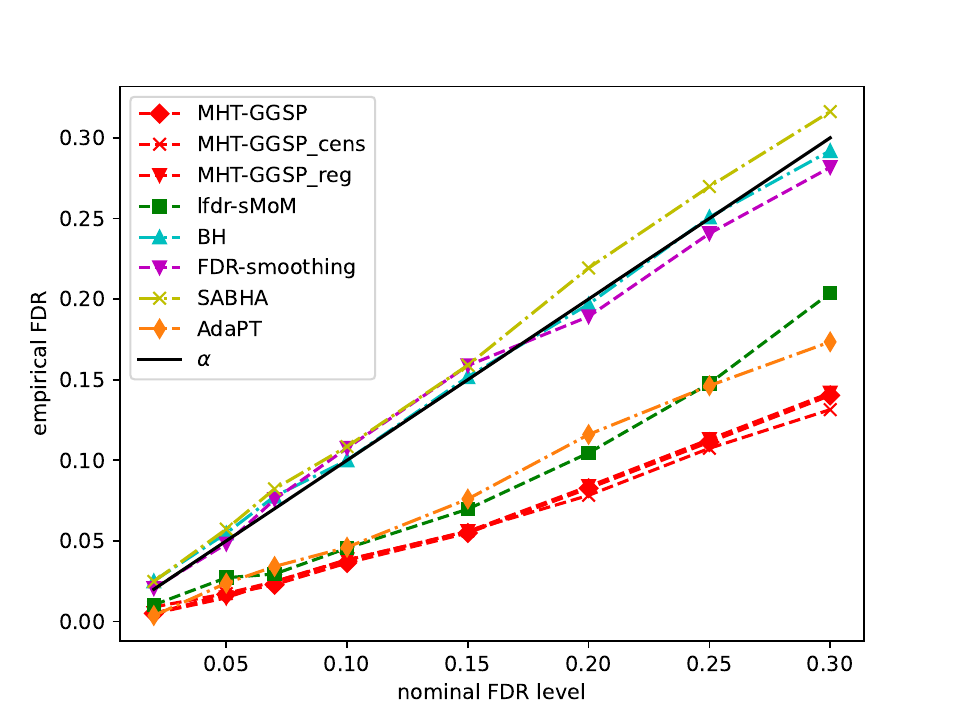}
		\caption{Empirical \gls{FDR} under noise level $1.25$}
		\label{fig:FDR_comm}
	\end{subfigure}
	\begin{subfigure}[b]{0.49\textwidth}
		\centering
		\includegraphics[width=0.95\linewidth, trim=0.55cm 0.3cm 1.5cm 1.4cm, clip]{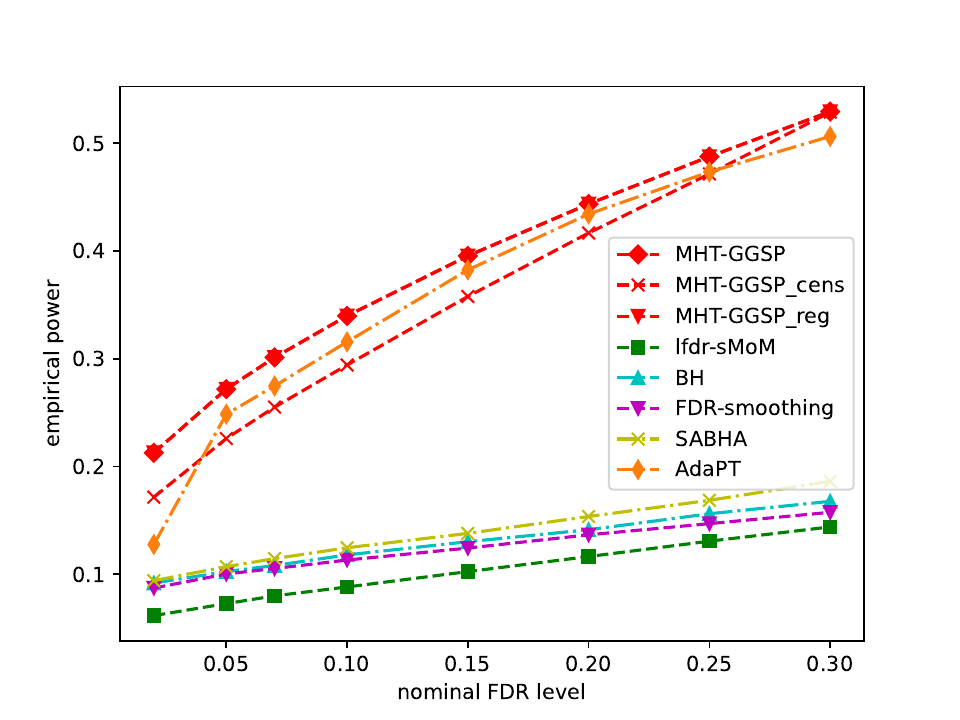}
		\caption{Empirical power under noise level $1.25$}
		\label{fig:pow_comm}
	\end{subfigure}
    \begin{subfigure}[b]{0.49\textwidth}
		\centering
		\includegraphics[width=0.95\linewidth, trim=0.55cm 0.3cm 1.5cm 1.4cm, clip]{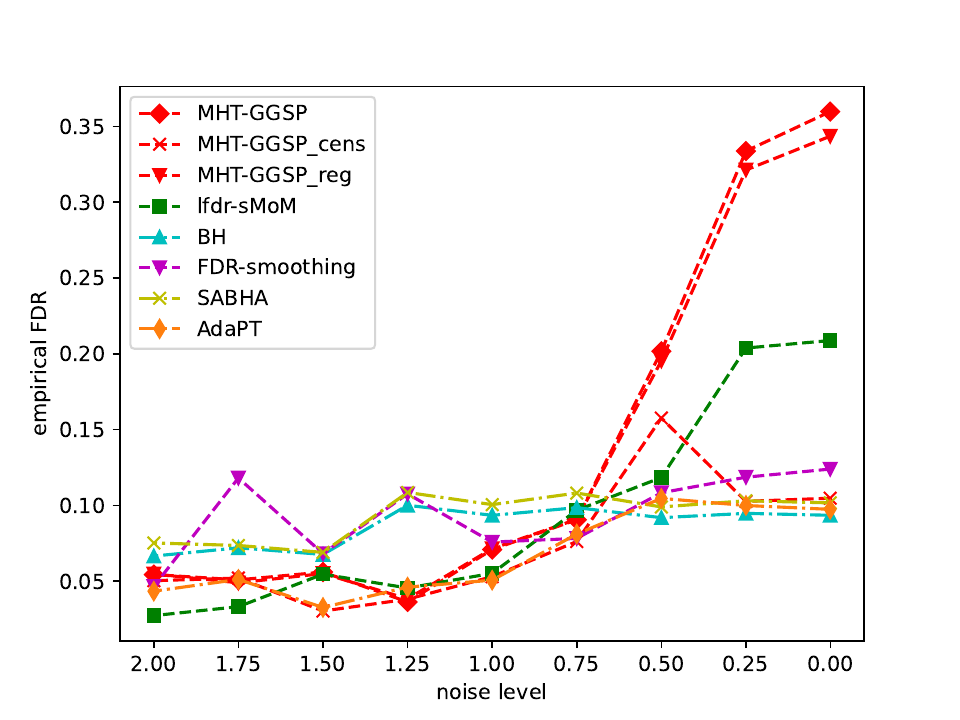}
		\caption{Empirical \gls{FDR} under nominal level $0.1$}
		\label{fig:FDR_vary_comm}
	\end{subfigure}
    \begin{subfigure}[b]{0.49\textwidth}
		\centering
		\includegraphics[width=0.95\linewidth, trim=0.55cm 0.3cm 1.5cm 1.4cm, clip]{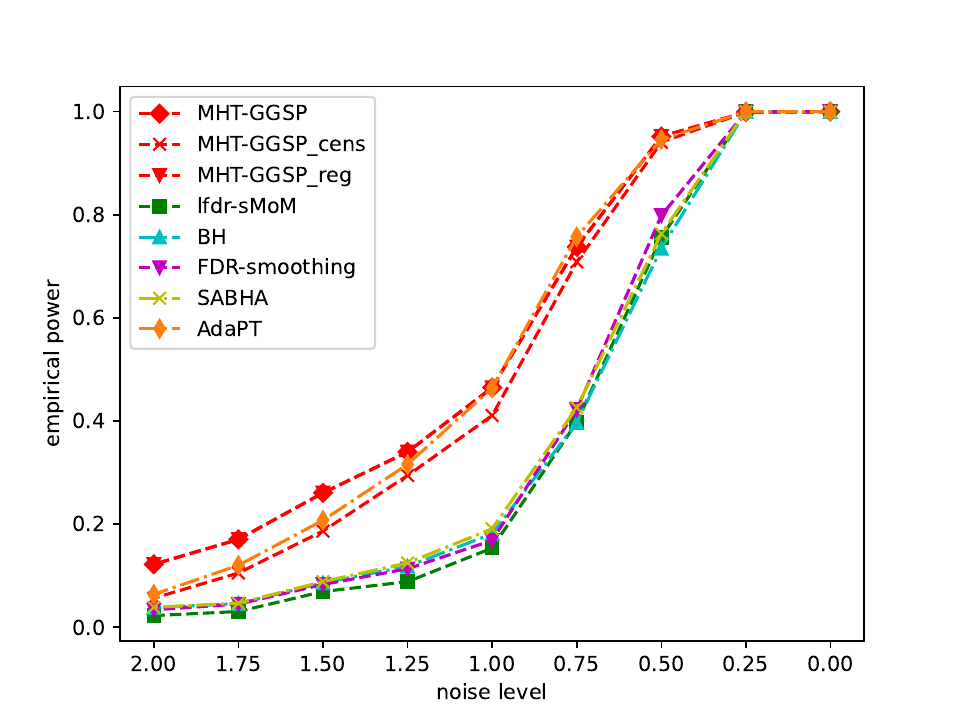}
		\caption{Empirical power under nominal level $0.1$}
		\label{fig:pow_vary_comm}
	\end{subfigure}

	\caption{\gls{FDR} and detection power under different nominal \gls{FDR} levels and noise levels. Each point is obtained by $20$ repetitions.}
	\label{fig:det_results}
\end{figure}

From \cref{fig:det_results}, we observe that the performance of MHT-GGSP and MHT-GGSP\_\text{reg} are almost identical. This indicates that the mean of the estimated $\widehat{\beta}\tuple$ by MHT-GGSP is very close to $\widetilde{\pi}_0$ estimated by Storey’s method. Therefore, the MHT-GGSP method provides reasonably accurate marginal estimates of the null probabilities. However, when the noise level is low, we observe \gls{FDR} inflation for both MHT-GGSP and MHT-GGSP\_\text{reg} (cf.\ \cref{fig:FDR_vary_comm}). This is due to the violation of the smooth variation assumption of the distribution and parameter underestimation, as previously mentioned. After introducing the censoring strategy in MHT-GGSP\_\text{cens}, this inflation is mitigated. However, in edge cases, such as when the noise level is $0.5$, this inflation cannot be entirely avoided. Additionally, since MHT-GGSP\_\text{cens} does not utilize the contextual information from the censored $p$-values, it is less powerful compared to MHT-GGSP and AdaPT when the noise level exceeds $0.75$.

Among the baseline methods, both lfdr-sMoM and FDR-smoothing utilize \gls{lfdr} similar to MHT-GGSP. These methods first estimate the densities under their respective empirical Bayesian frameworks, then compute \gls{lfdr}, and finally identify the largest set with an average \gls{lfdr} below the nominal level as $\altRegEst$. The lfdr-sMoM method assumes a traditional two-groups model, while FDR-smoothing assumes that the null proportion is inhomogeneous over the time-vertex graph. As shown in \cref{fig:pow_comm,fig:pow_vary_comm}, these methods are not as powerful as MHT-GGSP. This is because they both overlook the local variation of alternative $p$-value (or $z$-value) distributions, which is caused by the variation of signal energy over $\jointDom$, as illustrated in \cref{fig:p_dist_illustrate}.

The methods proportion-matching, SABHA, and AdaPT are extensions of the \gls{BH} method. The proportion-matching method has empirical \gls{FDR} greater than $0.5$ over all noise levels, thus is incomparable and not shown in \cref{fig:det_results}. Similar to FDR-smoothing, this method is designed under the assumption that the alternative $p$-value distribution is homogeneous over the sensor network. When this assumption does not hold, the \gls{FDR} is not guaranteed to be controlled at the nominal level. Besides, this method aims to approximate the performance of the \gls{BH} method, which is not as powerful as other methods. The SABHA method faces a similar issue: to achieve \gls{FDR} control, it imposes conditions on the Rademacher complexity of the feasible set of reciprocals of the $p$-values’ weights. These conditions may not always be met, leading to potential \gls{FDR} inflation, as seen in \cref{fig:FDR_comm}. By leveraging the sensor coordinate information, the AdaPT method achieves good power while maintaining \gls{FDR} control at the nominal level. However, as observed in \cref{fig:det_results}, AdaPT is not as powerful as MHT-GGSP when the noise level exceeds $0.75$. In summary, when the $p$-value distribution varies smoothly over $\jointDom$, MHT-GGSP demonstrates state-of-the-art performance by effectively utilizing this property.

\section{Conclusion}\label{sect:conc}

We have proposed a novel method for conducting network multiple hypothesis tests on the joint domain of the vertex set and the measure space of each vertex signal. Our approach models the distributions of $p$-values over the joint domain as parameterized by a random generalized graph signal, allowing for inhomogeneity across the joint domain. By utilizing this estimator in conjunction with the \gls{lfdr} based approach, we can control the \gls{FDR} at a specified nominal level in the asymptotic setting. This provides a powerful tool for accurately identifying significant hypotheses in complex datasets with an underlying graph structure. This approach also has the potential to be extended for online and distributed implementation over the underlying network.

\appendices

\section{Preliminaries: GGSP}\label[Appendix]{sect:prelim:GGSP}

Compared to the aforementioned traditional \gls{GSP} framework (cf.\ \cref{sect:stat_model}) where the vertex observations are scalars, the \gls{GGSP} framework considers the vertex observation as a function from $L^2(\calT)$ on each vertex. Specifically, a generalized graph signal $x$ is defined as the following map \cite{JiTay:J19}:
\begin{align*}
x: \calV &\to L^2(\calT) \\
v &\mapsto x(v)(\cdot).
\end{align*}
This map can be identified with the following element $x'\in L^2(\calV\times\calT)$:
\begin{align*}
x': \calV\times\calT &\to \Real \\
(v,t) &\mapsto x(v)(t). 
\end{align*}
In this paper, we refer to $L^2(\calV\times\calT)$ as the space of generalized graph signals.

Suppose the space $L^2(\calT)$ has an orthonormal basis $\set{\psi_k\given k=1,2,\dots}$ which also admits different smoothness over $\calT$. For example, when $\calT=[-\pi,\pi]$, the set of functions
\begin{align}\label{eq:tri_basis}
\set*{\ofrac{\sqrt{2\pi}}\Ind[-\pi,\pi]}\bigcup\set*{\ofrac{\sqrt{\pi}}\sin(kt),\ofrac{\sqrt{\pi}}\cos(kt)\given k=1,2,\dots}
\end{align}
forms an orthonormal basis for $L^2([-\pi,\pi])$. We also assume that $\set{\psi_k}$ are listed in decreasing order of smoothness. Write $\bphi_k=(\phi_k(v))_{v\in\calV}\in\Real^N$. It can be shown that the set of functions $\set{\phi_{k_1}(v)\psi_{k_2}(t)\given k_1=1,\dots,N; k_2=1,2,\dots}$ forms an orthonormal basis of $L^2(\calV\times\calT)$. A generalized graph signal $x$ is bandlimited if 
\begin{align*}
x(v,t) = \sum_{k_1=1}^{K_1}\sum_{k_2=1}^{K_2} \xi_{k_1,k_2}\phi_{k_1}(v)\psi_{k_2}(t),
\end{align*}
where $K_1,K_2<\infty$. By a linear combination of the smoothest basis vectors on $L^2(\calV\times\calT)$, this equation yields a smooth generalized graph signal. In this paper, we assume that $\calT$ is compact. By assuming a discrete topology on $\calV$ and the product topology on $\calV\times\calT$, it can be verified that $\calV\times\calT$ is compact. In the proofs, we use $\calJ$ to denote $\jointDom$ for simplicity.

\section{Proof of \cref{thm:opt_thres}}

We adopt a similar procedure to the proof of \cite[Theorem 2]{LeiFit:J18}. Specifically, we first rewrite problem \cref{eq:opt_mpow} as a convex problem, and then utilize \gls{KKT} condition to prove the final result.
For any $(v,t)\in\samReg$, we use $\FpNulNodeEpoch$ and $\FPAltNodeEpoch$ to denote the \glspl{cdf} of $\fPNulNodeEpoch$ and $\fPAltNodeEpoch$, respectively.  
First we simplify problem \cref{eq:opt_mpow} by introducing the following notations: 
\begin{align*}
a_0(\boldsymbol{s}) &:= \sum_{\tuple\in\samReg}\bbE[\idc{\p\leq s_{\tuple},\HTrue = \HNul}\given\signal(\samReg)] = \sum_{\tuple\in\samReg}\FpNulNodeEpoch[s_{\tuple}]\pi_0\circ\signalNodeEpoch, \\
a_1(\boldsymbol{s}) &:= \sum_{\tuple\in\samReg}\bbE[\idc{\p\leq s_{\tuple},\HTrue=\HAlt}\given\signal(\samReg)] = \sum_{\tuple\in\samReg} \FPAltNodeEpoch[s_{\tuple}](1-\pi_0\circ\signalNodeEpoch). 
\end{align*}
Note that the denominator of \cref{eq:def_mpow} does not depend on $h$. Therefore, problem \cref{eq:opt_mpow} is equivalent to 
\begin{align*}
&\max_{\boldsymbol{s}\in[0,1]^\numTests} a_1(\boldsymbol{s}) \\
\ST &\frac{a_0(\boldsymbol{s})}{a_0(\boldsymbol{s})+a_1(\boldsymbol{s})} \leq \alpha.
\end{align*}
This can be further simplified as 
\begin{align}
\begin{aligned}\label{eq:mfdr_Q}
&\min_{\boldsymbol{s}\in[0,1]^\numTests} -a_1(\boldsymbol{s}) \\
\ST & -\alpha a_1(\boldsymbol{s}) + (1-\alpha)a_0(\boldsymbol{s}) \leq 0.
\end{aligned}
\end{align}
From the monotonicity of $\fPNodeEpoch$ and $\fPAltNodeEpoch$, we know that $-a_1(\boldsymbol{s})$ and $-\alpha a_1(\boldsymbol{s}) + (1-\alpha)a_0(\boldsymbol{s})$ are convex in $\boldsymbol{s}$. Hence, \cref{eq:mfdr_Q} is a convex optimization problem. Besides, since the feasible region is compact and $\FpNulNodeEpoch$, $\FPAltNodeEpoch$ are continuous, we know that problem \cref{eq:mfdr_Q} has a global optimal solution. Next, we verify Slater's condition and then apply \gls{KKT} condition. Slater's condition requires that there exists an $\boldsymbol{s}$ such that $s_{\tuple}\in (0,1)$ for all $\tuple\in\jointDom$, and $-\alpha a_1(\boldsymbol{s}) + (1-\alpha)a_0(\boldsymbol{s}) < 0$. To find such $\boldsymbol{s}$, let $g(\boldsymbol{s}) := -\alpha a_1(\boldsymbol{s}) + (1-\alpha)a_0(\boldsymbol{s})$. According to the assumption, we suppose that there exists $p_0\in(0,1)$ and $(v_0, t_0)\in\samReg$ such that $\lfdr(p_0; \signal_{(v_0, t_0)})< \alpha$. We define the constant 
\begin{align*}
c_1 = -\alpha \fPAltNodeEpoch[p_0][v_0][t_0](1-\pi_0\circ\signal_{(v_0,t_0)}) + (1-\alpha)\fPNulNodeEpoch[p_0][v_0][t_0]\pi_0\circ\signal_{(v_0,t_0)}
\end{align*}
which is negative by assumption. According to \cref{cond:conts_pdf} in \cref{asp:identify}, we suppose $\fPNulNodeEpoch$ is uniformly bounded by $c_2$ on $[0,1]\times\calJ$.
Given an arbitrary $\epsilon>0$, let $\boldsymbol{s}\in[0,1]^\numTests$ be such that:
\begin{align*}
0&<s_{\tuple}\leq \min(p_0,\frac{\epsilon}{2(\numTests-1)c_2}), \forall \tuple\in\samReg-\set{(v_0, t_0)}, \\
-\frac{\epsilon}{c_1}&<s_{(v_0,t_0)}\leq p_0.
\end{align*}
Note that $g$ is differentiable on the interior of $[0,1]^\numTests$. Therefore, using mean value theorem \cite[Corollary 10.2.9]{Tao:16}, we know that there exists an $\boldsymbol{s}'$ such that 
\begin{align*}
&g(\boldsymbol{s}) - g(\bzero) = \nabla g(\boldsymbol{s}')\T \boldsymbol{s}\\
&= (-\alpha \fPAltNodeEpoch[s'_{(v_0,t_0)}][v_0][t_0](1-\pi_0\circ\signal_{(v_0,t_0)}) + (1-\alpha)\fPNulNodeEpoch[s'_{(v_0,t_0)}][v_0][t_0]\pi_0\circ\signal_{(v_0,t_0)})s_{(v_0,t_0)} + \\ 
&\sum_{\tuple\neq(v_0,t_0)} (-\alpha \fPAltNodeEpoch[s'_{\tuple}](1-\pi_0\circ\signal_{\tuple}) + (1-\alpha)\fPNulNodeEpoch[s'_{\tuple}]\pi_0\circ\signal_{\tuple})s_{\tuple} \\
&\leq (-\alpha \fPAltNodeEpoch[p_0][v_0][t_0](1-\pi_0\circ\signal_{(v_0,t_0)}) + (1-\alpha)\fPNulNodeEpoch[p_0][v_0][t_0]\pi_0\circ\signal_{(v_0,t_0)})s_{(v_0,t_0)} + \\
&\sum_{\tuple\neq(v_0,t_0)} (-\alpha \fPAltNodeEpoch[p_0](1-\pi_0\circ\signal_{\tuple}) + (1-\alpha)\fPNulNodeEpoch[p_0]\pi_0\circ\signal_{\tuple})s_{\tuple} \\
&\leq c_1s_{(v_0,t_0)} + \sum_{\tuple\neq(v_0,t_0)}c_2s_{\tuple}<-\frac{\epsilon}{2}<0.
\end{align*}
Here $\boldsymbol{s}' = \mu' \boldsymbol{s}$, $\mu'\in(0,1)$. Notice that $g(\bzero)=0$, and $\boldsymbol{s}$ is in the interior of $[0,1]^\numTests$, hence Slater's condition is satisfied. In this case, $\boldsymbol{s}^*$ is the optimal solution if and only if it satisfies the \gls{KKT} conditions (cf. \cite[pp.244]{Boyd:16}), one of which is that the gradient of Lagrangian $L(\boldsymbol{s},\mu) = -a_1(\boldsymbol{s}) + \mu g(\boldsymbol{s})$ is zero. Therefore, we have
\begin{align*}
\frac{\partial}{\partial s_{\tuple}}L(\boldsymbol{s}^*,\mu^*)&= -\fPAltNodeEpoch[s^*_{\tuple}](1-\pi_0\circ\signal_{\tuple}) + \mu(-\alpha \fPAltNodeEpoch[s^*_{\tuple}](1-\pi_0\circ\signal_{\tuple}) + \\
&(1-\alpha)\fPNulNodeEpoch[s^*_{\tuple}]\pi_0\circ\signal_{\tuple})=0.
\end{align*}
Rearranging the terms we obtain
\begin{align*}
\lfdr(s^*_{\tuple};\signal_{\tuple}) = \frac{1+\mu\alpha}{1+\mu},
\end{align*}
which completes the proof.

\section{Proof of \cref{thm:unif_conv_para}}

This theorem follows directly from the consistency of \gls{MLE}. Given $\bXi$, from \cref{eq:pdf_pvt|gam}, the \gls{pdf} of $\p,\tuple$ can be written as 
\begin{align*}
&f_{\mathrm{P},\tupleRV}(p,(v,t)\mid\signal) = \pi_0\Big(\sum_{k_1=1}^{K_1}\sum_{k_2=1}^{K_2} \upxi_{k_1,k_2}\cdot\phi_{k_1}(v)\psi_{k_2}(t)\Big)f_0(p;(v,t))\rho(v,t) \\
&+ \Big(1-\pi_0\Big(\sum_{k_1=1}^{K_1}\sum_{k_2=1}^{K_2} \upxi_{k_1,k_2}\cdot\phi_{k_1}(v)\psi_{k_2}(t)\Big)\Big) f_1\Big(p;\sum_{k_1=1}^{K_1}\sum_{k_2=1}^{K_2} \upxi_{k_1,k_2}\cdot\phi_{k_1}(v)\psi_{k_2}(t)\Big)\rho(v,t).
\end{align*}
For every $p$ and $(v,t)$, this \gls{pdf} is continuous \gls{wrt} $\bXi$. By \cref{cond:identi_pvt} in \cref{asp:bandlimit}, we know that different values of $\bXi$ lead to different joint distributions of $\p,\tupleRV$. According to \cite[Theorem 9.9]{Kee:10} and \cref{lem:det_to_rand}, $\widehat{\bXi}\convp\bXi$ given $\signal$. By \cref{cond:cont_basis} in \cref{asp:bandlimit} and the compactness of $\calJ$, suppose $\set{\abs{\phi_{k_1}(v)\psi_{k_2}(t)}}$ are uniformly bounded by $b$. We have
\begin{align*}
\abs{\widehat{\signal}(v,t) - \signal(v,t)}
&\leq \sum_{k_1=1}^{K_1}\sum_{k_2=1}^{K_2} \abs{\widehat{\upxi}_{k_1,k_2} - \upxi_{k_1,k_2}} \cdot\abs{\phi_{k_1}(v)\psi_{k_2}(t)}\\
&\leq b\cdot\sum_{k_1=1}^{K_1}\sum_{k_2=1}^{K_2} \abs{\widehat{\upxi}_{k_1,k_2} - \upxi_{k_1,k_2}}.
\end{align*}
Note that the right-hand side does not depend on $(v,t)$ and converges to zero in probability given $\bXi$. This concludes the proof.

\section{Proof of \cref{thm:limit_FDR,thm:limit_pow}}

To prove \cref{thm:limit_FDR,thm:limit_pow}, we first prove the following lemmas. First, \cref{lem:det_to_rand} allows us to prove results under a probability measure conditioned on $\signal$. Second, \cref{lem:limit_conts_cdf_ratio,lem:equal_infs,lem:exs_small_ratio,lem:L1-sum_to0} indicate that under \cref{asp:identify,asp:bandlimit,asp:regularity}, the conditions \cite[C1-C3]{CaoChenZhang:J22} hold. Then \cref{thm:limit_FDR,thm:limit_pow} can be obtained by using \cite[Theorem 2.3, Theorem 3.5]{CaoChenZhang:J22} and \cref{lem:det_to_rand}. We denote the set of all sample paths of $\signal(v,t)$ as $\Gamma$. Suppose $\Gamma$ is a measure space with $\sigma$-algebra $\calF_\Gamma$, then the random $\signal$ induces a probability measure $\bbP_\signal$ on $(\Gamma, \calF_\Gamma)$.

\begin{LemmaA}\label{lem:det_to_rand}
 Let $\bm{\upiota}$ represent $(\p,\tuple,\HTrue)_{\tuple\in\samReg}$ for simplicity. Denote $\widetilde{\calJ}:= [0,1]\times\calJ\times\set{\HNul,\HAlt}$. Let $\calF_\numTests$ be the $\sigma$-algebra of the space $\widetilde{\calJ}^\numTests\times\Gamma$. Define 
\begin{align*}
\mu:\Gamma\times\calF_\numTests&\to [0,1]\\
(\gamma,A) 
&\mapsto \int_{\widetilde{\calJ}^\numTests}\Ind\set{((\iota_i)_{i=1}^\numTests,\gamma
)\in A}\prod\limits_{i=1}^\numTests f_{\rVar{P}, \rVar{H}, \tupleRV}(\iota_i\mid\gamma)\ud \iota_i.
\end{align*}
Then $\mu(\signal(\omega),A)$ is a \gls{r.c.d.} of $((\p,\tuple,\HTrue)_{\tuple\in\samReg},\signal)$ given $\signal$ \cite[Section 4.1.3]{Dur:19}.
\end{LemmaA}
\begin{proof}
We prove this lemma by definition of \gls{r.c.d.}. As a stochastic process, $\signal$ can be written as $\signal(\omega, (v,t))$. In this lemma, we abuse the notation to write $\signal(\omega, \cdot)$ as $\signal(\omega)$. First, we show that for each $A\in\calF_\numTests$, $\mu(\signal(\omega),A)$ is a version of $\bbP((\bm{\upiota},\signal)\in A\given\signal)$. Let $F\in\calF_\Gamma$. Then we have 
\begin{align*}
\bbE[\mu(\signal(\omega),A)\Ind \set{\signal(\omega)\in F}] 
&= \int_\Omega\int_{\widetilde{\calJ}^\numTests}\Ind\set{((\iota_i)_{i=1}^\numTests,\signal(\omega)
)\in A}\prod\limits_{i=1}^\numTests f_{\rVar{P}, \rVar{H}, \tupleRV}(\iota_i\mid\signal(\omega))\ud \iota_i \Ind \set{\signal(\omega)\in F} \ud\bbP(\omega)\\
&= \int_\Gamma\int_{\widetilde{\calJ}^\numTests}\Ind\set{((\iota_i)_{i=1}^\numTests,\gamma
)\in A}\prod\limits_{i=1}^\numTests f_{\rVar{P}, \rVar{H}, \tupleRV}(\iota_i\mid\gamma)\ud \iota_i \Ind \set{\gamma\in F} \ud\bbP_\signal(\gamma) \\
&= \bbE[\Ind\set{(\bm{\upiota},\signal)\in A}\Ind\set{\signal\in F}].
\end{align*}
Second, it can be shown that for each $\gamma\in\Gamma$, $\mu(\gamma, \cdot)$ is a probability measure on $\widetilde{\calJ}^\numTests\times\Gamma$, hence $\mu(\gamma, A)$ is a \gls{r.c.d.} of $\bm{\upiota}$ given $\signal = \gamma$.
\end{proof}

From \cref{lem:det_to_rand} we know that when proving results under a probability measure conditioned on $\signal$, we can regard $\signal$ as being fixed. In the lemmas and proofs below, for clarity, we may not write variables dependent on $\bXi$ as random variables when we are proving under the probability measure conditioned on $\bXi$.

\begin{LemmaA}\label{lem:limit_conts_cdf_ratio}
Suppose \cref{asp:identify,asp:bandlimit,asp:regularity} hold. Under the probability measure conditioned on $\bXi$, we have
\begin{align*}
\rVar{d}_{0,\numTests}(\eta)&\convp \rVar{d}_0(\eta), \\
\rVar{d}_{1,\numTests}(\eta)&\convp \rVar{d}_1(\eta), \\
\rVar{d}'_{1,\numTests}(\eta)&\convp \rVar{d}_1(\eta)
\end{align*}
as $\numTests\to\infty$, where $\rVar{d}_0(\eta)$ and $\rVar{d}_1(\eta)$ are continuous \gls{wrt} $\eta$.
\end{LemmaA}
\begin{proof}
Note that \cref{eq:def:D_1M,eq:def:V_M,eq:def:D_0M} are summations of \gls{iid} random variables (conditioned on $\bXi$, same as below), and each random variable takes value in $[0,1]$, thus has finite expectation and variance. 
Then by \gls{WLLN}, we know that $\rVar{d}_{0,\numTests}(\eta)\convp\bbE[\idc{\lfdr(\pRVNodeEpochRV;\signal(\nodeRV,\epochRV;\bXi))\leq\eta} \given \bXi]$, which proves the existence of $\rVar{d}_0(\eta)$. Besides, since $\bbE[\rVar{d}_{1,\numTests}(\eta)\given \bXi] = \bbE[\rVar{d}'_{1,\numTests}(\eta)\given \bXi]$, we know that the limits of $\rVar{d}_{1,\numTests}(\eta)$ and $\rVar{d}'_{1,\numTests}(\eta)$ are the same, both equal to 
\begin{align*}
\bbE[\idc{\rVar{H}_{(\nodeRV,\epochRV)} = \HNul}\Ind\set{\lfdr(\pRVNodeEpochRV;\signal(\nodeRV,\epochRV;\bXi))\leq\eta}\given \bXi],
\end{align*}
which then proves the existence of $\rVar{d}_1(\eta)$.

Next, we prove the continuity of $\rVar{d}_0(\eta)$ and $\rVar{d}_1(\eta)$. For $\rVar{d}_0(\eta)$, note that it is a \gls{cdf}, so it is right continuous \cite[Theorem 1.5.1]{HogMcCra:19}. To verify that it is left continuous, we calculate the limit
\begin{align*}
\lim_{\eta'\to\eta^{-}} \abs{\rVar{d}_0(\eta) - \rVar{d}_0(\eta')} &= \lim_{\eta'\to\eta^{-}}\bbE[\Ind\set{\eta'<\lfdr(\pRVNodeEpochRV;\signal(\nodeRV,\epochRV;\bXi))\leq\eta}\given \bXi] \\
&=\bbE[\lim_{\eta'\to\eta^{-}}\Ind\set{\eta'<\lfdr(\pRVNodeEpochRV;\signal(\nodeRV,\epochRV;\bXi))\leq\eta}\given \bXi]\\
&=\bbE[\Ind\set{\lfdr(\pRVNodeEpochRV;\signal(\nodeRV,\epochRV;\bXi))=\eta}\given \bXi]\\
&=\int_{\calJ}\int_{(0,1]} \Ind\set{\lfdr(p;\signal(v,t;\bXi))=\eta} \fP[p\mid\signal(v,t;\bXi))]\rho(v,t)\ud p  \ud(v,t),
\end{align*}
where the second equality is derived by \gls{DCT} \cite[Theorem 1.13]{SteSha:05}. By \cref{cond:monotone_pdf} in \cref{asp:regularity}, $\lfdr(p;\signal(v,t;\bXi))$ strictly increases in $p$. Therefore, for every fixed $(v,t)$, $\abs{\set{\lfdr(p;\signal((v,t);\bXi))=\eta}}\leq1$. Since $\fP[p\mid\signal(v,t;\bXi)]$ is continuous in $p$, we know that the last line in the above equation equals zero, hence $\lim\limits_{\eta'\to\eta^-} \abs{\rVar{d}_0(\eta) - \rVar{d}_0(\eta')}=0$, i.e., $\rVar{d}_0(\eta)$ is continuous on $[0,1]$. The continuity of $\rVar{d}_1(\eta)$ can be proved similarly.
\end{proof}

\begin{LemmaA}\label{lem:equal_infs}
Let $\tilde{F}_0(\eta;\signal(v,t;\bXi))$ be the \gls{cdf} of $\lfdr(\pRVNodeEpochRV;\signal(\nodeRV,\epochRV;\bXi))$ given $\tupleRV=\tuple$ and $\HTrueRV = \HNul$. Specifically, it can be calculated as
\begin{align*}
\tilde{F}_0(\eta;\signal(v,t;\bXi)):= \int_{(0,1]} \Ind\set{\lfdr(p;\signal(v,t;\bXi))\leq\eta} \fPNulNodeEpoch\ud p.
\end{align*}
Then the following quantities are equal:
\begin{align*}
\eta_{\infty,1}'&:=\inf\set{\eta\given \rVar{d}_1(\eta)>0}, \\
\eta_{\infty,2}'&:=\inf\set{\eta\given\rho(\set{(v,t)\given\tilde{F}_0(\eta;\signal(v,t;\bXi))>0})>0}, \\
\eta_{\infty,3}' &:= \min_{(v,t)\in\calJ} \chi(v,t;\bXi).
\end{align*}
We denote $\eta_\infty':= \eta_{\infty,1}' = \eta_{\infty,2}' = \eta_{\infty,3}'$.
\end{LemmaA}
\begin{proof}
We prove the inequalities $\eta_{\infty,2}'\leq \eta_{\infty,1}'$, $\eta_{\infty,3}'\leq \eta_{\infty,2}'$ and $\eta_{\infty,1}'\leq \eta_{\infty,3}'$. Note that $\rVar{d}_1(\eta)$ can be calculated by
\begin{align}\label{eq:D_1}
\rVar{d}_1(\eta) &= \int_{\calJ}\int_{(0,1]}\sum_{H\in\set{\HNul,\HAlt}}\idc{\lfdr(p;\signal(v,t;\bXi))\leq\eta,H = \HNul}f_{\rVar{P}, \rVar{H}, \tupleRV}(p,H,\tuple\mid\signal)\ud p\ud\tuple \nn
&= \int_{\calJ}\pi_0\circ\signal(v,t;\bXi)\tilde{F}_0(\eta;\signal(v,t;\bXi))\rho\tuple\ud\tuple.
\end{align}
For any $\eta$ such that $\rVar{d}_1(\eta)>0$, by \cref{eq:D_1} and \cref{cond:pi0_pos} in \cref{asp:regularity} we know that 
\begin{align}\label{eq:eta_pos_measure}
\rho(\set{(v,t)\given\tilde{F}_0(\eta;\signal(v,t;\bXi))>0})>0,
\end{align}
hence $\eta_{\infty,2}'\leq \eta_{\infty,1}'$. Next, we want to prove that for any $\eta$ satisfying \cref{eq:eta_pos_measure}, $\eta\geq \min\limits_{(v,t)\in\calJ}\chi(v,t;\bXi)$. We prove this claim by contradiction. Suppose there exists $\eta_{\infty,2}'\leq\eta< \min\limits_{(v,t)\in\calJ}\chi(v,t;\bXi)=\eta_{\infty,3}'$. Then for any $(v,t)\in\calJ$, $\eta$ is below the range of $\lfdr(\cdot;\signal(v,t;\bXi))$, hence $\tilde{F}_0(\eta;\signal(v,t;\bXi)) = 0$, which contradicts with \cref{eq:eta_pos_measure}. Hence, $\eta_{\infty,3}'\leq \eta_{\infty,2}'$. 

Finally, we prove $\eta_{\infty,1}'\leq \eta_{\infty,3}'$ by contradiction. Suppose $\eta_{\infty,3}'< \eta_{\infty,1}'$. Then there exists $\eta\in(\eta_{\infty,3}', \eta_{\infty,1}')$. Since $\eta<\eta_{\infty,1}'$, we know that $\rVar{d}_1(\eta)=0$. On the other hand, since $\eta>\eta_{\infty,3}'$, and $\chi(v,t;\bXi)$ is continuous (cf.\ \cref{cond:conts_limit_lfdr} in \cref{asp:regularity}), we know that $\set{(v,t)\given\chi(v,t;\bXi)< \eta}$ is a non-empty open set. Since $\rho$ is a strictly positive measure, we know that this set has a positive measure \gls{wrt} $\rho$. Besides, note that $\chi(v,t;\bXi)< \eta$ implies $\tilde{F}_0(\eta;\signal(v,t;\bXi))>0$, therefore 
\begin{align*}
\rho(\set{(v,t)\given\tilde{F}_0(\eta;\signal(v,t;\bXi))>0})>0,
\end{align*}
and 
\begin{align*}
\rVar{d}_1(\eta) \geq \int_{\calJ}\Ind\set{(v,t)\given\tilde{F}_0(\eta;\signal(v,t;\bXi))>0} \tilde{F}_0(\eta;\signal(v,t;\bXi))\min_{\tuple\in\calJ}\pi_0\circ\signal(v,t;\bXi) \ud \rho\tuple >0,
\end{align*}
which leads to a contradiction with $\rVar{d}_1(\eta)=0$. This completes the proof.
\end{proof}

\begin{LemmaA}\label{lem:exs_small_ratio}
Under \cref{asp:identify,asp:regularity}, there exists $\eta_\infty\in(0,1]$ such that $\rVar{r}(\eta_\infty)<\alpha$.
\end{LemmaA}
\begin{proof}
According to \cref{lem:equal_infs}, for any $\eta>\eta_\infty'$, we have $\rVar{d}_0(\eta)\geq \rVar{d}_1(\eta)>0$ and 
\begin{align*}
\frac{\rVar{d}_1(\eta)}{\rVar{d}_0(\eta)} = \frac{\bbE[\lfdr(\pRVNodeEpochRV;\signal(\nodeRV,\epochRV;\bXi))\Ind\set{\lfdr(\pRVNodeEpochRV;\signal(\nodeRV,\epochRV;\bXi))\leq\eta}\given \bXi]}{\bbE[\Ind\set{\lfdr(\pRVNodeEpochRV;\signal(\nodeRV,\epochRV;\bXi))\leq\eta}\given \bXi]}\leq\eta.
\end{align*}
Therefore, it suffices to prove that $\eta_\infty'<\alpha$. By \cref{lem:equal_infs}, $\eta_\infty' = \min\limits_{(v,t)\in\calJ}\chi(v,t;\bXi)$. According to \cref{cond:small_lfdr} in \cref{asp:regularity}, there exists a $(v,t)$ such that $\chi(v,t;\bXi)<\alpha$. This completes the proof.
\end{proof}

\begin{LemmaA}\label{lem:L1-sum_to0}
Under \cref{asp:identify,asp:bandlimit,asp:regularity}, we have
\begin{align*}
\ofrac{\numTests}\sum_{\tuple\in\samReg}\abs*{\lfdr(\p;\signal(v,t;\widehat{\bXi})) - \lfdr(\p;\signal(v,t;\bXi))}\convp 0,
\end{align*}
under the probability measure conditioned on $\bXi$.
\end{LemmaA}
\begin{proof} First we observe that
\begin{align*}
&\ofrac{\numTests}\sum_{\tuple\in\samReg}\abs*{\lfdr(\p;\signal(v,t;\widehat{\bXi})) - \lfdr(\p;\signal(v,t;\bXi))}\\
&\leq \ofrac{\numTests}\sum_{\tuple\in\samReg}\sup_{(u,s)\in\calJ}\abs*{\lfdr(\p;\signal(u,s;\widehat{\bXi})) - \lfdr(\p;\signal(u,s;\bXi))}.
\end{align*}
By Markov's inequality, it suffices to prove that the conditional expectation of the right-hand side tends to zero, i.e.,  

\begin{align}\label{eq:esupto0}
\bbE[\sup_{(u,s)\in\calJ}\abs*{\lfdr(\pRVNodeEpochRV;\signal(u,s;\widehat{\bXi})) - \lfdr(\pRVNodeEpochRV;\signal(u,s;\bXi))}\given \bXi]\to0.
\end{align}

To this end, we find an upper bound for $\abs*{\lfdr(\pRVNodeEpochRV;\signal(u,s;\widehat{\bXi})) - \lfdr(\pRVNodeEpochRV;\signal(u,s;\bXi))}$:
\begin{align*}
&\abs*{\lfdr(\pRVNodeEpochRV;\signal(u,s;\widehat{\bXi})) - \lfdr(\pRVNodeEpochRV;\signal(u,s;\bXi))} \\
&= \abs*{\frac{\pi_0\circ\signal(u,s;\widehat{\bXi})\fPNulNodeEpoch[\pRVNodeEpochRV][u][s]}{\fP[\pRVNodeEpochRV\mid\signal(u,s;\widehat{\bXi})]} - \frac{\pi_0\circ\signal(u,s,\bXi)\fPNulNodeEpoch[\pRVNodeEpochRV][u][s]}{\fP[\pRVNodeEpochRV\mid\signal(u,s;\bXi)]}}\\
&\leq \abs*{\frac{\pi_0\circ\signal(u,s;\widehat{\bXi})\fPNulNodeEpoch[\pRVNodeEpochRV][u][s]}{\fP[\pRVNodeEpochRV\mid\signal(u,s;\widehat{\bXi})]} - \frac{\pi_0\circ\signal(u,s;\widehat{\bXi})\fPNulNodeEpoch[\pRVNodeEpochRV][u][s]}{\fP[\pRVNodeEpochRV\mid\signal(u,s;\bXi)]}} \\
&+ \abs*{\frac{\pi_0\circ\signal(u,s;\widehat{\bXi})\fPNulNodeEpoch[\pRVNodeEpochRV][u][s]}{\fP[\pRVNodeEpochRV\mid\signal(u,s;\bXi)]} - \frac{\pi_0\circ\signal(u,s;\bXi)\fPNulNodeEpoch[\pRVNodeEpochRV][u][s]}{\fP[\pRVNodeEpochRV\mid\signal(u,s;\bXi)]}}.
\end{align*}
We denote the first and second terms on the right-hand side as $\mathrm{term1}$ and $\mathrm{term2}$. In the rest of this proof, we aim to find upper bounds for them. Since $\fPNulNodeEpoch$ is continuous on $[0,1]\times\calJ$, we assume that it is upper bounded by $b_1$. To find a lower bound for $\fP[\pRVNodeEpochRV\mid\signal(u,s;\bXi)]$, using \cref{cond:pi0_pos,cond:f0_pos} in \cref{asp:regularity}, we have
\begin{align*}
\fP[p\mid\signal(u,s;\bXi)] &\geq \pi_0\circ\signal(u,s;\bXi)\fPNulNodeEpoch[p][u][s]\geq \pi_0\circ\signal(u,s;\bXi)\fPNulNodeEpoch[\ofrac{2}][u][s]>0,\forall p>\ofrac{2},\\
\fP[p\mid\signal(u,s;\bXi)] &\geq (1-\pi_0\circ\signal(u,s;\bXi))\fPAlt[p;\signal(u,s;\bXi)] \\
&\geq (1-\pi_0\circ\signal(u,s;\bXi))\fPAlt[\ofrac{2};\signal(u,s;\bXi)]>0,\forall p\leq\ofrac{2}.
\end{align*}
Let 
\begin{align*}
b_2:=\min\parens*{\min_{(u,s)\in\calJ}\pi_0\circ\signal(u,s;\bXi)\fPNulNodeEpoch[\ofrac{2}][u][s]>0,\min_{(u,s)\in\calJ}(1-\pi_0\circ\signal(u,s;\bXi))\fPAlt[\ofrac{2};\signal(u,s;\bXi)]},
\end{align*}
then $\fP[p\mid\signal(u,s;\bXi)]\geq b_2>0$ for all $p\in(0,1]$ and $(u,s)\in\calJ$.
Then we have
\begin{align*}
\mathrm{term1}&\leq\min\parens*{b_1\abs*{\ofrac{\fP[\pRVNodeEpochRV\mid\signal(u,s;\widehat{\bXi})]} - \ofrac{\fP[\pRVNodeEpochRV\mid\signal(u,s;\bXi)]}}, 1+\frac{b_1}{b_2}}:=\min(\mathrm{term1}',b_3),
\end{align*}
thus for any $\delta\in(0,1)$, we have
\begin{align*}
\mathrm{term1}
&\leq \mathrm{term1}'\Ind\set{\pRVNodeEpochRV\geq\delta} + b_3\Ind\set{\pRVNodeEpochRV\in(0,\delta)}.
\end{align*}
We first consider the event $\pRVNodeEpochRV\geq\delta$. Since $\signal(u,s;\bXi)$ is continuous on $(u,s,\bXi)$ and $\calJ\times\calK$ is compact, the image $\signal(\calJ;\calK):=\set{\signal(u,s;\bXi)\given (u,s)\in\calJ,\bXi\in\calK}$ is compact.  According to \cref{cond:conts_pd} in \cref{asp:regularity}, $\dfrac{\partial f'_{\rVar{P}}}{\partial\zeta}$ is continuous on $[\delta,1]\times\signal(\calJ;\calK)\times\calJ$, so there exists $b_4(\delta)>0$ so that $\abs*{\dfrac{\partial f'_{\rVar{P}}}{\partial\zeta}(p\mid\signal(u,s;\bXi),(u,s))}\leq b_4(\delta)$ for any $p\in[\delta,1]$, $(u,s)\in\calJ$ and $\bXi\in\calK$. Then there exists $\bXi'= c\bXi + (1-c)\widehat{\bXi}$, $c\in(0,1)$ such that on the event $\pRVNodeEpochRV\geq\delta$, we have
\begin{align*}
\mathrm{term1}' &= b_1\ofrac{\fP[\pRVNodeEpochRV\mid\signal(u,s;\bXi')]^2}\abs*{\frac{\partial f'_{\rVar{P}}}{\partial\zeta}(\pRVNodeEpochRV\mid\signal(u,s;\bXi'),(u,s))}\abs*{\signal(u,s;\bXi) - \signal(u,s;\widehat{\bXi})} \\
&\leq \frac{b_1}{b_2^2}b_4(\delta)\sup_{(u',s')\in\calJ}\abs{\signal(u',s';\bXi) - \signal(u',s';\widehat{\bXi})}.
\end{align*}
Hence, $\mathrm{term1}$ can be upper bounded as 
\begin{align}\label{eq:T1bound}
\mathrm{term1}\leq \frac{b_1}{b_2^2}b_4(\delta)\sup_{(u,s)\in\calJ}\abs{\signal(u,s;\bXi) - \signal(u,s;\widehat{\bXi})} + b_3\Ind\set{\pRVNodeEpochRV\in(0,\delta)}.
\end{align}
Next, we find an upper bound for $\mathrm{term2}$. Define $b_5:=\min\limits_{(u,s)\in\calJ}\fPNulNodeEpoch[1][u][s]>0$ and 
\begin{align*}
\mathrm{term2}':=\abs*{\fP[1\mid\signal(u,s;\widehat{\bXi})] - \fP[1\mid\signal(u,s;\bXi)]}.
\end{align*}
\begin{align*}
\mathrm{term2} &= \frac{\fPNulNodeEpoch[\pRVNodeEpochRV][u][s]}{\fP[\pRVNodeEpochRV\mid\signal(u,s;\bXi)]} \abs{\pi_0\circ\signal(u,s;\widehat{\bXi}) - \pi_0\circ\signal(u,s;\bXi)}\\
&\leq\frac{b_1}{b_2} \abs{\pi_0\circ\signal(u,s;\widehat{\bXi}) - \pi_0\circ\signal(u,s;\bXi)}\\
&\leq\frac{b_1}{b_2}\ofrac{\fPNulNodeEpoch[1][u][s]}\abs*{\fP[1\mid\signal(u,s;\widehat{\bXi})] - \fP[1\mid\signal(u,s;\bXi)]}.\\
&\leq\frac{b_1}{b_2b_5}\mathrm{term2}'.
\end{align*}
Using a similar argument as before, there exists $\bXi''$ such that
\begin{align*}
\mathrm{term2}' = \abs*{\frac{\partial f'_{\rVar{P}}}{\partial\zeta}(1\mid\signal(u,s,\bXi''),(u,s))}\abs*{\signal(u,s;\widehat{\bXi}) - \signal(u,s;\bXi)}\leq b_4(\delta)\sup_{(u,s)\in\calJ} \abs*{\signal(u,s;\widehat{\bXi}) - \signal(u,s;\bXi)},
\end{align*}
hence 
\begin{align}\label{eq:T2bound}
\mathrm{term2}\leq \frac{b_1b_4(\delta)}{b_2b_5}\sup_{(u,s)\in\calJ} \abs*{\signal(u,s;\widehat{\bXi}) - \signal(u,s;\bXi)}. 
\end{align}
Combining \cref{eq:T1bound,eq:T2bound} and redefining the constants, we obtain
\begin{align}
\begin{aligned}\label{eq:ub_by_op1}
&\sup_{(u,s)\in\calJ}\abs*{\lfdr(\pRVNodeEpochRV;\signal(u,s;\widehat{\bXi})) - \lfdr(\pRVNodeEpochRV;\signal(u,s;\bXi))} \\
&\leq b_6(\delta)\sup_{(u,s)\in\calJ} \abs*{\signal(u,s;\widehat{\bXi}) - \signal(u,s;\bXi)} + b_7\Ind\set{\pRVNodeEpochRV\in(0,\delta)}.
\end{aligned}
\end{align}
By \cref{thm:unif_conv_para}, $\sup\limits_{(u,s)\in\calJ} \abs*{\signal(u,s;\widehat{\bXi}) - \signal(u,s;\bXi)}\convp 0$. 
For any $\epsilon_1,\epsilon_2>0$, we first choose $\delta>0$ such that $\bbP(\pRVNodeEpochRV\in(0,\delta)\given \bXi)<\dfrac{\epsilon_2}{2}$. Then when $\numTests$ is large enough, we have 
\begin{align*}
\bbP(b_6(\delta)\sup_{(u,s)\in\calJ} \abs*{\signal(u,s;\widehat{\bXi}) - \signal(u,s;\bXi)}\geq\epsilon_1\given \bXi)<\frac{\epsilon_2}{2}.
\end{align*}
Combining these results, we obtain 
\begin{align*}
\bbP(\sup_{(u,s)\in\calJ}\abs*{\lfdr(\pRVNodeEpochRV;\signal(u,s;\widehat{\bXi})) - \lfdr(\pRVNodeEpochRV;\signal(u,s;\bXi))}\geq2\epsilon_1\given \bXi)<\epsilon_2,
\end{align*}
i.e., $\sup\limits_{(u,s)\in\calJ}\abs*{\lfdr(\pRVNodeEpochRV;\signal(u,s;\widehat{\bXi})) - \lfdr(\pRVNodeEpochRV;\signal(u,s;\bXi))}\convp0$ conditioned on $\bXi$. Note that this is a sequence of uniformly bounded random variables, hence uniformly integrable. Then by applying \cite[Theorem 8.16]{Kee:10}, 
we obtain \cref{eq:esupto0}. This completes the proof.
\end{proof}
\begin{proof}[Proof of \cref{thm:limit_FDR}]
We first prove the result under the probability measure conditioned on $\bXi$. The result then follows by taking expectation over $\bXi$ in the last step. According to \cref{lem:det_to_rand,lem:limit_conts_cdf_ratio,lem:exs_small_ratio,lem:L1-sum_to0} and \cite[Theorem 3.2]{CaoChenZhang:J22}, we know that\footnote{Note that the definition of $\FDR_m$ in \cite{CaoChenZhang:J22} is different from the definition \cref{eq:def_fdr} in this paper.} 
\begin{align}\label{eq:E_orac_contr}
\uplim_{\numTests\to\infty}\bbE[\frac{\rVar{d}'_{1,\numTests}(\widehat{\upeta}_\numTests)}{\max(\rVar{d}_{0,\numTests}(\widehat{\upeta}_\numTests),\ofrac{\numTests})}\given \bXi]\leq\alpha.
\end{align}
The result \cref{eq:E_orac_contr} is based on the true values of \gls{lfdr}. We use this result to provide an upper bound for $\FDR_{\samReg}(h)=\bbE[\dfrac{\widehat{\rVar{d}'}_{1,\numTests}(\widehat{\upeta}_\numTests)}{\max(\widehat{\rVar{d}}_{0,\numTests}(\widehat{\upeta}_\numTests),\ofrac{\numTests})}]$: 
\begin{align*}
&\abs*{\frac{\widehat{\rVar{d}'}_{1,\numTests}(\widehat{\upeta}_\numTests)}{\max(\widehat{\rVar{d}}_{0,\numTests}(\widehat{\upeta}_\numTests),\ofrac{\numTests})} - \frac{\rVar{d}'_{1,\numTests}(\widehat{\upeta}_\numTests)}{\max(\rVar{d}_{0,\numTests}(\widehat{\upeta}_\numTests),\ofrac{\numTests})}} \\
&\leq \abs*{\frac{\widehat{\rVar{d}'}_{1,\numTests}(\widehat{\upeta}_\numTests)}{\max(\widehat{\rVar{d}}_{0,\numTests}(\widehat{\upeta}_\numTests),\ofrac{\numTests})} - \frac{\widehat{\rVar{d}'}_{1,\numTests}(\widehat{\upeta}_\numTests)}{\max(\rVar{d}_{0,\numTests}(\widehat{\upeta}_\numTests),\ofrac{\numTests})}} +\abs*{\frac{\widehat{\rVar{d}'}_{1,\numTests}(\widehat{\upeta}_\numTests)}{\max(\rVar{d}_{0,\numTests}(\widehat{\upeta}_\numTests),\ofrac{\numTests})} - \frac{\rVar{d}'_{1,\numTests}(\widehat{\upeta}_\numTests)}{\max(\rVar{d}_{0,\numTests}(\widehat{\upeta}_\numTests),\ofrac{\numTests})}}\\
&\leq \frac{\widehat{\rVar{d}'}_{1,\numTests}(\widehat{\upeta}_\numTests)}{\max(\widehat{\rVar{d}}_{0,\numTests}(\widehat{\upeta}_\numTests),\ofrac{\numTests})}\ofrac{\max(\rVar{d}_{0,\numTests}(\widehat{\upeta}_\numTests),\ofrac{\numTests})}\abs*{\max(\rVar{d}_{0,\numTests}(\widehat{\upeta}_\numTests),\ofrac{\numTests})-\max(\widehat{\rVar{d}}_{0,\numTests}(\widehat{\upeta}_\numTests),\ofrac{\numTests})}\\
&+\ofrac{\max(\rVar{d}_{0,\numTests}(\widehat{\upeta}_\numTests),\ofrac{\numTests})}\abs*{\widehat{\rVar{d}'}_{1,\numTests}(\widehat{\upeta}_\numTests) - \rVar{d}'_{1,\numTests}(\widehat{\upeta}_\numTests)}\\
&\leq \ofrac{\max(\rVar{d}_{0,\numTests}(\widehat{\upeta}_\numTests),\ofrac{\numTests})}\parens*{\abs*{\max(\rVar{d}_{0,\numTests}(\widehat{\upeta}_\numTests),\ofrac{\numTests})-\max(\widehat{\rVar{d}}_{0,\numTests}(\widehat{\upeta}_\numTests),\ofrac{\numTests})}+ \abs*{\widehat{\rVar{d}'}_{1,\numTests}(\widehat{\upeta}_\numTests) - \rVar{d}'_{1,\numTests}(\widehat{\upeta}_\numTests)}}
\end{align*}
When $\widehat{\upeta}_\numTests\geq \eta_\infty$, $\rVar{d}_{0,\numTests}(\widehat{\upeta}_\numTests)\geq \rVar{d}_{0,\numTests}(\eta_\infty)$. Combining \cite[(26)]{CaoChenZhang:J22} and \cite[Lemma 8.1]{CaoChenZhang:J22}, we know that $\sup\limits_{\eta\in[0,1]}\abs*{\rVar{d}_{0,\numTests}(\eta)-\rVar{d}_0(\eta)}\convp0$ and $\sup\limits_{\eta\geq\eta_\infty}\abs*{\rVar{d}_{0,\numTests}(\eta)-\widehat{\rVar{d}}_{0,\numTests}(\eta)}\convp0$ (conditioned on $\bXi$, the same below), hence
\begin{align*}
\max(\rVar{d}_{0,\numTests}(\eta_\infty),\ofrac{\numTests})\convp\max(\rVar{d}_0(\eta_\infty),0) = \rVar{d}_0(\eta_\infty)>0.
\end{align*}
Since the function $\max(a,\ofrac{\numTests})$ is Lipschitz continuous in $a$ with Lipschitz constant $1$, we have
\begin{align*}
\abs*{\max(\rVar{d}_{0,\numTests}(\widehat{\upeta}_\numTests),\ofrac{\numTests})-\max(\widehat{\rVar{d}}_{0,\numTests}(\widehat{\upeta}_\numTests),\ofrac{\numTests})}
&\leq \abs{\rVar{d}_{0,\numTests}(\widehat{\upeta}_\numTests)-\widehat{\rVar{d}}_{0,\numTests}(\widehat{\upeta}_\numTests)}\\
&\leq \sup\limits_{\eta\geq\eta_\infty}\abs*{\rVar{d}_{0,\numTests}(\eta)-\widehat{\rVar{d}}_{0,\numTests}(\eta)}\convp0.
\end{align*}
Besides, notice that
\begin{align*}
\abs*{\widehat{\rVar{d}'}_{1,\numTests}(\eta) - \rVar{d}'_{1,\numTests}(\eta)}
&\leq\ofrac{\numTests}\sum_{\tuple\in\samReg} \abs*{\Ind\set{\lfdr(\p;\signal(v,t;\bXi))\leq\eta}-\Ind\set{\lfdr(\p;\signal(v,t;\widehat{\bXi}))\leq\eta}}\idc{\HTrue = \HNul}\\
&\leq\ofrac{\numTests}\sum_{\tuple\in\samReg} \abs*{\Ind\set{\lfdr(\p;\signal(v,t;\bXi))\leq\eta}-\Ind\set{\lfdr(\p;\signal(v,t;\widehat{\bXi}))\leq\eta}}.
\end{align*}
In the proof of \cite[Lemma 8.4]{CaoChenZhang:J22}, we know that 
\begin{align*}
\sup_{\eta\geq\eta_\infty} \ofrac{\numTests}\sum_{\tuple\in\samReg} \abs*{\Ind\set{\lfdr(\p;\signal(v,t;\bXi))\leq\eta}-\Ind\set{\lfdr(\p;\signal(v,t;\widehat{\bXi}))\leq\eta}}\convp0.
\end{align*}
Therefore,  $\abs*{\widehat{\rVar{d}'}_{1,\numTests}(\widehat{\upeta}_\numTests) - \rVar{d}'_{1,\numTests}(\widehat{\upeta}_\numTests)}\leq\sup\limits_{\eta\geq\eta_\infty}\abs*{\widehat{\rVar{d}'}_{1,\numTests}(\eta) - \rVar{d}'_{1,\numTests}(\eta)}\convp0$. Combining the above results, we obtain
\begin{align}\label{eq:smalldiff_bigeta}
\abs*{\frac{\widehat{\rVar{d}'}_{1,\numTests}(\widehat{\upeta}_\numTests)}{\max(\widehat{\rVar{d}}_{0,\numTests}(\widehat{\upeta}_\numTests),\ofrac{\numTests})} - \frac{\rVar{d}'_{1,\numTests}(\widehat{\upeta}_\numTests)}{\max(\rVar{d}_{0,\numTests}(\widehat{\upeta}_\numTests),\ofrac{\numTests})}}\Ind\set{\widehat{\upeta}_\numTests\geq\eta_\infty}\convp0.
\end{align}
When $\widehat{\upeta}_\numTests<\eta_\infty$, 
\begin{align*}
\abs*{\frac{\widehat{\rVar{d}'}_{1,\numTests}(\widehat{\upeta}_\numTests)}{\max(\widehat{\rVar{d}}_{0,\numTests}(\widehat{\upeta}_\numTests),\ofrac{\numTests})} - \frac{\rVar{d}'_{1,\numTests}(\widehat{\upeta}_\numTests)}{\max(\rVar{d}_{0,\numTests}(\widehat{\upeta}_\numTests),\ofrac{\numTests})}}\Ind\set{\widehat{\upeta}_\numTests<\eta_\infty}\leq 2 \Ind\set{\widehat{\upeta}_\numTests<\eta_\infty}.
\end{align*}
According to the proof of \cite[Theorem 2.3]{CaoChenZhang:J22}, we know that $\bbP(\widehat{\upeta}_\numTests\geq\eta_\infty\given \bXi)\to1$, hence the \gls{RHS} tends to $0$ in conditional probability. Combining this with \cref{eq:smalldiff_bigeta}, we obtain that
\begin{align*}
\abs*{\frac{\widehat{\rVar{d}'}_{1,\numTests}(\widehat{\upeta}_\numTests)}{\max(\widehat{\rVar{d}}_{0,\numTests}(\widehat{\upeta}_\numTests),\ofrac{\numTests})} - \frac{\rVar{d}'_{1,\numTests}(\widehat{\upeta}_\numTests)}{\max(\rVar{d}_{0,\numTests}(\widehat{\upeta}_\numTests),\ofrac{\numTests})}}\convp0.
\end{align*}
This sequence of random variables is uniformly bounded by $2$, hence is uniformly integrable. By \cite[Theorem 8.16]{Kee:10} we know that its expectation tends to $0$. Finally, by the existing result \cref{eq:E_orac_contr}, the result follows by
\begin{align*}
\uplim_{\numTests\to\infty}\bbE[\abs*{\frac{\widehat{\rVar{d}'}_{1,\numTests}(\widehat{\upeta}_\numTests)}{\max(\widehat{\rVar{d}}_{0,\numTests}(\widehat{\upeta}_\numTests),\ofrac{\numTests})}}\given \bXi]
&\leq\uplim_{\numTests\to\infty}\bbE[\abs*{\frac{\widehat{\rVar{d}'}_{1,\numTests}(\widehat{\upeta}_\numTests)}{\max(\widehat{\rVar{d}}_{0,\numTests}(\widehat{\upeta}_\numTests),\ofrac{\numTests})} - \frac{\rVar{d}'_{1,\numTests}(\widehat{\upeta}_\numTests)}{\max(\rVar{d}_{0,\numTests}(\widehat{\upeta}_\numTests),\ofrac{\numTests})}}\given \bXi] \\
&+ \bbE[\abs*{\frac{\rVar{d}'_{1,\numTests}(\widehat{\upeta}_\numTests)}{\max(\rVar{d}_{0,\numTests}(\widehat{\upeta}_\numTests),\ofrac{\numTests})}}\given \bXi]
\leq\alpha.
\end{align*}
Further applying reverse Fatou's lemma, we obtain
\begin{align*}
\uplim_{\numTests\to\infty}\bbE[\bbE[\abs*{\frac{\widehat{\rVar{d}'}_{1,\numTests}(\widehat{\upeta}_\numTests)}{\max(\widehat{\rVar{d}}_{0,\numTests}(\widehat{\upeta}_\numTests),\ofrac{\numTests})}}\given\bXi]]
\leq\bbE[\uplim_{\numTests\to\infty}\bbE[\abs*{\frac{\widehat{\rVar{d}'}_{1,\numTests}(\widehat{\upeta}_\numTests)}{\max(\widehat{\rVar{d}}_{0,\numTests}(\widehat{\upeta}_\numTests),\ofrac{\numTests})}}\given\bXi]]
\leq\alpha,
\end{align*}
which concludes the proof.
\end{proof}

\begin{proof}[Proof of \cref{thm:limit_pow}] 
We first prove the result under the probability measure conditioned on $\bXi$. The result then follows by taking expectation over $\bXi$.  We first prove that $\rVar{r}(\eta)$ strictly increases in $\eta\in(\eta_\infty',1)$. Let $\eta_1,\eta_2\in(\eta_\infty',1)$ satisfy $\eta_1<\eta_2$. Then
\begin{align*}
\rVar{r}(\eta_2) - \rVar{r}(\eta_1) &= \frac{\rVar{d}_1(\eta_2)}{\rVar{d}_0(\eta_2)} - \frac{\rVar{d}_1(\eta_1)}{\rVar{d}_0(\eta_1)} \\
&= \frac{\rVar{d}_1(\eta_2)\rVar{d}_0(\eta_1)-\rVar{d}_1(\eta_1)\rVar{d}_0(\eta_1)}{\rVar{d}_0(\eta_1)\rVar{d}_0(\eta_2)} + \frac{\rVar{d}_1(\eta_1)\rVar{d}_0(\eta_1)-\rVar{d}_1(\eta_1)\rVar{d}_0(\eta_2)}{\rVar{d}_0(\eta_1)\rVar{d}_0(\eta_2)}.
\end{align*}
The numerator can be rewritten as
\begin{align*}
&\rVar{d}_1(\eta_2)\rVar{d}_0(\eta_1)-\rVar{d}_1(\eta_1)\rVar{d}_0(\eta_1) + \rVar{d}_1(\eta_1)\rVar{d}_0(\eta_1)-\rVar{d}_1(\eta_1)\rVar{d}_0(\eta_2) \\
&= \rVar{d}_0(\eta_1)(\rVar{d}_1(\eta_2)-\rVar{d}_1(\eta_1)) - \rVar{d}_1(\eta_1)(\rVar{d}_0(\eta_2)-\rVar{d}_0(\eta_1)) \\
&=\bbE[\Ind\set{\lfdr(\pRVNodeEpochRV;\signal(\nodeRV,\epochRV;\bXi))\leq\eta_1}\given \bXi]\\
&\cdot\bbE[\lfdr(\pRVNodeEpochRV;\signal(\nodeRV,\epochRV;\bXi))\Ind\set{\eta_1<\lfdr(\pRVNodeEpochRV;\signal(\nodeRV,\epochRV;\bXi))\leq\eta_2}\given \bXi] \\
&- \bbE[\lfdr(\pRVNodeEpochRV;\signal(\nodeRV,\epochRV;\bXi))\Ind\set{\lfdr(\pRVNodeEpochRV;\signal(\nodeRV,\epochRV;\bXi))\leq\eta_1}\given \bXi] \\
&\cdot\bbE[\Ind\set{\eta_1<\lfdr(\pRVNodeEpochRV;\signal(\nodeRV,\epochRV;\bXi))\leq\eta_2}\given \bXi] \\
&> \eta_1\bbE[\Ind\set{\lfdr(\pRVNodeEpochRV;\signal(\nodeRV,\epochRV;\bXi))\leq\eta_1}\given \bXi]\bbE[\Ind\set{\eta_1<\lfdr(\pRVNodeEpochRV;\signal(\nodeRV,\epochRV;\bXi))\leq\eta_2}\given \bXi] \\
&- \eta_1\bbE[\Ind\set{\lfdr(\pRVNodeEpochRV;\signal(\nodeRV,\epochRV;\bXi))\leq\eta_1}\given \bXi]\bbE[\Ind\set{\eta_1<\lfdr(\pRVNodeEpochRV;\signal(\nodeRV,\epochRV;\bXi))\leq\eta_2}\given \bXi]=0.
\end{align*}
Since $\eta_1,\eta_2\in(\eta_\infty',1)$, $\rVar{d}_0(\eta_2)\geq \rVar{d}_0(\eta_1)>0$, the denominator $\rVar{d}_0(\eta_1)\rVar{d}_0(\eta_2)$ is positive, i.e, $\rVar{r}(\eta)$ strictly increases in $\eta\in(\eta_\infty',1)$. Therefore, due to the continuity of $\rVar{d}_1$ and $\rVar{d}_0$, we know that for any small $\epsilon>0$, $\rVar{r}(\upeta_0-\epsilon)<\alpha$ holds. By \gls{WLLN}, we know that $\ofrac{\numTests}\sum\limits_{\tuple\in\samReg}\idc{\HTrue = \HNul}\convp \upkappa_0$, and $\ofrac{\numTests}\sum\limits_{\tuple\in\samReg} \Ind\set{\HTrue=\HAlt,\lfdr(\p;\signal(v,t;\bXi))\leq\eta}\convp \rVar{d}_2(\eta)$, both conditioned on $\bXi$. Hence by applying \cite[Theorem 3.5]{CaoChenZhang:J22}, we obtain the desired result when conditioned on $\bXi$. The result under unconditional probability measure can be obtained by taking expectation over $\bXi$.
\end{proof}

\bibliographystyle{IEEEtran}
\bibliography{IEEEabrv,StringDefinitions,refs}

\end{document}